\documentclass[11pt]{article}

\usepackage[T1]{fontenc}
\usepackage[utf8]{inputenc}
\usepackage{lmodern}
\usepackage{microtype}
\usepackage{amsmath,amssymb,amsthm,mathtools}
\usepackage{enumitem}
\usepackage{needspace}
\usepackage{xcolor}
\usepackage{tikz}
\usetikzlibrary{arrows.meta}
\usepackage[a4paper,margin=30mm]{geometry}
\usepackage[unicode,hidelinks]{hyperref}
\usepackage[nameinlink,capitalize,noabbrev]{cleveref}

\newtheorem{theorem}{Theorem}[section]
\newtheorem{proposition}[theorem]{Proposition}
\newtheorem{lemma}[theorem]{Lemma}
\newtheorem{corollary}[theorem]{Corollary}
\theoremstyle{definition}
\newtheorem{definition}[theorem]{Definition}
\newtheorem{example}[theorem]{Example}
\theoremstyle{remark}
\newtheorem{remark}[theorem]{Remark}

\DeclareMathOperator{\Span}{span}
\DeclareMathOperator{\im}{im}
\DeclareMathOperator{\rank}{rank}
\DeclareMathOperator{\rt}{rt}
\newcommand{\C}{\mathbb C}
\newcommand{\Q}{\mathbb Q}
\newcommand{\one}{\mathbf 1}
\newcommand{\muRoots}{\boldsymbol\mu}
\newcommand{\calL}{\mathcal L}
\newcommand{\calM}{\mathcal M}
\newcommand{\calN}{\mathcal N}
\newcommand{\calZ}{\mathcal Z}
\newcommand{\set}[1]{\left\{#1\right\}}
\newcommand{\abs}[1]{\left|#1\right|}

\tikzset{
  oc state/.style={
    circle,
    draw,
    thick,
    minimum size=7mm,
    inner sep=0pt
  },
  oc cycle state/.style={oc state,fill=black!10},
  oc transition/.style={draw,thick,-{Stealth[length=2mm]}}
}

\title{The \v{C}ern\'y Conjecture for One-Cluster Automata\\
       via Annular Spectral Descent}
\author{Yinfeng Zhu\\[0.4em]\small Independent Researcher}
\date{Draft of 21 July 2026}

\hypersetup{
  pdftitle={The Černý Conjecture for One-Cluster Automata via Annular Spectral Descent},
  pdfauthor={Yinfeng Zhu},
  pdfsubject={Synchronizing one-cluster automata},
  pdfkeywords={Černý conjecture, synchronizing automata, one-cluster automata}
}

\begin{document}

\maketitle

\begin{abstract}
We prove the \v{C}ern\'y conjecture for synchronizing one-cluster
  automata.  More precisely, let a synchronizing automaton with state set
  $Q$, $\abs Q=n$, have a letter $a$ whose functional digraph has a unique
  cycle $C$ of length $m$, and let $\ell$ be the least nonnegative integer
  for which $a^\ell$ maps $Q$ onto $C$.
  Assume $\ell\ge1$.  For every nonempty proper subset $S\subset C$, we
  prove that there is a word $w$ of length at most $n$ such that
  $wa^\ell$ maps more than $\abs S$ states of $C$ into $S$.
This proves the positive-level part of a conjecture of Kisielewicz,
Kowalski, and Szyku\l a concerning relative extending words for one-cluster
automata.  The resulting reset word has length at most
\[
  (m-1)(n-1)+m\ell\le(n-1)^2.
\]
For every $n\ge4$, we construct a strongly connected binary example with
$m=2$, $\ell=n-2$, and reset threshold $3n-5$, so the parameter-dependent
bound $(m-1)(n-1)+m\ell$ is sharp.
The upper-bound proof uses finite-dimensional linear algebra; the sharpness
lower bounds are combinatorial.

\par\smallskip
\noindent\textbf{AI use statement.}
The proof was obtained through interaction with OpenAI Codex
(GPT-5.6 Sol, ultra mode) and verified by the author.
\end{abstract}



\section{Introduction}\label{sec:introduction}

A complete deterministic finite automaton $A$ with state set $Q$ is
\emph{synchronizing} if some word maps all its states to one state.  Put
$n=\abs Q$.  If $A$ is synchronizing, the least possible length of such a
word is its reset threshold, denoted by $\rt(A)$.  The \v{C}ern\'y conjecture
asserts that every synchronizing automaton with $n$ states has reset threshold
at most
$(n-1)^2$; the classical examples of \v{C}ern\'y show that this number would
be best possible~\cite[Lemma~1]{Cerny1964}.  As of 14 January 2026,
Volkov's living list recorded the general conjecture as neither proved nor
disproved~\cite[Section~1.2]{Volkov2026}; see also
\cite[Section~3.1]{Volkov2022Survey}.

Throughout, words act on states on the right:
$q\mathbin{\cdot}(uv)=(q\mathbin{\cdot}u)\mathbin{\cdot}v$.  For a subset
$X\subseteq Q$ and a word $w$, we write $X\mathbin{\cdot}w$ for its image and
$Xw^{-1}=\{q\in Q:q\mathbin{\cdot}w\in X\}$ for its full preimage.  These
conventions are formalized in \cref{sec:preliminaries}.

An automaton with state set $Q$ is \emph{one-cluster} if it has a letter $a$
whose functional digraph has connected underlying undirected graph;
equivalently, that functional digraph has exactly one directed cycle.  We
call $a$ a \emph{one-cluster letter}, call its unique directed cycle the
\emph{$a$-cycle}, and denote the vertex set of that cycle by $C$.  The
number of cycle states is denoted by $m=\abs C$.  The
\emph{level with respect to $a$} is
\[
  \ell=\min\set{j\ge0:Q\mathbin{\cdot}a^j=C}.
\]
Circular automata are the special case $\ell=0$; equivalently, $C=Q$ and
$a$ acts on $Q$ as a cyclic permutation.  Pin proved the conjecture for
circular automata with a prime number of states~\cite[Theorem~2]{Pin1978},
and Dubuc proved it for all circular
automata~\cite[Proposition~4.6]{Dubuc1998}.
Synchronizing one-cluster automata admit quadratic reset bounds.  In
particular, Carpi and D'Alessandro proved
\[
  \rt(A)\le 2n^2-4n+1-2(n-1)\ln(n/2)
\]
for the full class
~\cite[Proposition~5 and Corollary~1]{CarpiDAlessandro2013}; see also the
earlier bound in \cite[Proposition~2]{BealPerrin2009}.  Steinberg proved
the \v{C}ern\'y bound when the unique cycle has prime length
~\cite[Theorem~8]{Steinberg2011PrimeCycle}.

Kisielewicz, Kowalski, and Szyku\l a conjectured the following stronger
relative extending-word property
~\cite[Conjecture~1]{KisielewiczKowalskiSzykula2016}.  Every nonempty proper
subset $S\subset C$ should admit a word $w$ of length at most $n$ for which
the relative preimage $S(wa^\ell)^{-1}\cap C$ has cardinality greater than
$\abs S$.  They observed that this statement implies the \v{C}ern\'y bound
for the whole synchronizing one-cluster class.  The purpose of this paper is
to prove their conjecture at every positive level.  It follows that
\[
  \rt(A)\le (m-1)(n-1)+m\ell\le(n-1)^2.
\]

The refined estimate is attained at positive level in every order $n\ge4$.
In \cref{sec:sharpness}, we construct strongly connected binary automata with
$m=2$, $\ell=n-2$, and reset threshold $3n-5$, exactly the right-hand side of
the refined estimate.  The estimate is not attained by every positive-level
family.  For $n>2$, the one-cluster family
$\mathcal D_n^{\prime\prime}$ of Ananichev, Gusev, and Volkov has $m=n-1$,
$\ell=1$, and
$\rt(\mathcal D_n^{\prime\prime})=(n-1)(n-2)$
~\cite[Figure~8 and Theorem~6]{AnanichevGusevVolkov2013}.  Our estimate gives
$(n-1)^2$ for this family, which is larger by $n-1$.

The statement is cycle-relative: it counts only preimages lying in $C$.
It also bounds the prefix $w$, rather than the whole word $wa^\ell$.  These
two features distinguish it from ordinary extension and from Volkov's
$1$-extensibility.  The precise definitions and the relation among the
terminologies are given in \cref{def:relative-extension}.

Steinberg's prime-cycle proof uses centered characteristic vectors and an
ascending-chain argument within his averaging framework
~\cite[equation~(1), Lemma~2, Proposition~5, Lemma~6, and
Theorem~8]{Steinberg2011PrimeCycle}; see also the averaging lemma in
\cite[Lemma~2]{Steinberg2011}.  The present proof retains
that relative preimage framework but replaces the prime-cycle cyclotomic
irreducibility step by a Fitting decomposition and a roots-of-unity
cancellation that also applies to composite cycle lengths.  The filtered
spaces below refine a construction of Dubuc; the precise specialization to
the circular case is recorded when those spaces are defined.

Unpublished workshop slides supplied by Costa outline joint work with
Steinberg for one-cluster automata having a single off-cycle state, that is,
$n=m+1$ and $\ell=1$~\cite{CostaSteinberg2025Workshop}.  Like our proof,
their argument filters transition rows by the number of letters different
from $a$.  The outline identifies two difficulties at higher level: the
nilpotent part of the $a$-transition can have longer chains, and the required
rows may no longer admit sufficiently short word representatives.  The
Fitting decomposition, the two generation bounds, and annular descent
overcome these difficulties.

The obstruction beyond level one is the subspace of complex row vectors
annihilated by a power of the transition operator of the one-cluster letter.
After spectral projectors are inserted into a word product, a term may enter
this nilpotent subspace and later return to the invertible Fitting summand.  A
simplex of all exponent tuples up to one total degree does not cancel all such
terms.  We instead omit a precisely determined lower simplex.  The remaining
total degrees form an annular interval; when its lower endpoint is positive,
the interval has exactly $m$ consecutive integers.  This width permits the
roots-of-unity recurrence, and the shifted lower endpoint removes its
uncancelled boundary coefficient.  The exact endpoints are defined in
\cref{sec:annular}.

The remainder of this paper is organized as follows.
\Cref{sec:preliminaries} fixes the automata-theoretic and matrix notation,
states the main results, and introduces a five-state binary automaton of level
two that serves as a running example.  \Cref{sec:fitting,sec:filtered,sec:annular}
develop the linear-algebraic machinery: the Fitting decomposition, short
actual-word generators for filtered spaces, and the annular descent identity.
\Cref{sec:short-extension-proof} converts that identity into a short relative
extending word by positive-coefficient induction.  \Cref{sec:singleton} gives
the singleton start, and \cref{sec:reset-proof} completes the reset accounting
and proves the main theorem.  \Cref{sec:remarks} records structural features of
the annular argument, while \cref{sec:sharpness} constructs, in every order at
least four, a family attaining the refined estimate.  The positive-level proof
is self-contained; only the circular boundary case invokes Dubuc's theorem.
The Costa--Steinberg slides are cited as an unpublished proof outline and are
not used as a proof dependency.

\section{Automata, matrices, and the main results}\label{sec:preliminaries}

This section defines the word action, transition matrices, and basic
one-cluster parameters, and states the two results to be proved.  The
counting identity below is the bridge between preimage growth and linear
algebra.

Let $A=\langle Q,\Sigma,\delta\rangle$ be a complete deterministic finite
automaton: $Q$ and $\Sigma$ are finite nonempty sets and
$\delta:Q\times\Sigma\to Q$ is defined on every pair.  The elements of $Q$
are the \emph{states}, the elements of $\Sigma$ are the \emph{input letters},
and $\Sigma^*$ is the free monoid of all finite words over $\Sigma$ under
concatenation.  Extend $\delta$ to
$Q\times\Sigma^*$ by
\[
  \delta(q,\varepsilon)=q,
  \qquad
  \delta(q,wg)=\delta(\delta(q,w),g)
  \quad(q\in Q,\ w\in\Sigma^*,\ g\in\Sigma).
\]
Write $q\mathbin{\cdot}w=\delta(q,w)$.  The empty word
is denoted by $\varepsilon$, and $\abs w$ denotes the length of $w$.  For
$S\subseteq Q$, define the image and the full preimage by
\[
  S\mathbin{\cdot}w=\set{q\mathbin{\cdot}w:q\in S},
  \qquad
  Sw^{-1}=\set{q\in Q:q\mathbin{\cdot}w\in S}.
\]
Here $w^{-1}$ is full-preimage notation; it does not assert that the
transformation induced by $w$ is invertible.  We have
$S\mathbin{\cdot}(uv)=(S\mathbin{\cdot}u)\mathbin{\cdot}v$ and
$S(uv)^{-1}=(Sv^{-1})u^{-1}$.

Fix an ordering of $Q$, and let $P_w$ be the matrix of the transformation
$q\mapsto q\mathbin{\cdot}w$ in the standard state basis:
\[
  (P_w)_{p,q}=\begin{cases}
  1,&p\mathbin{\cdot}w=q,\\
  0,&p\mathbin{\cdot}w\ne q.
  \end{cases}.
\]
Thus $P_{uv}=P_uP_v$, and every row of $P_w$ contains exactly one entry
equal to $1$.  For $S\subseteq Q$, let $[S]$ denote its
\emph{characteristic row vector}, and let $\one=[Q]^{\mathsf T}$ be the
all-one column.  Then
\[
  P_w[S]^{\mathsf T}=[Sw^{-1}]^{\mathsf T},
  \qquad P_w\one=\one,
\]
and consequently, for $D,S\subseteq Q$,
\begin{equation}\label{eq:counting}
  [D]P_w[S]^{\mathsf T}=\abs{D\cap Sw^{-1}}.
\end{equation}
A word $w$ is a \emph{reset word} (also called a \emph{synchronizing word}) if
$\abs{Q\mathbin{\cdot}w}=1$.  The automaton is \emph{synchronizing} if it
has a reset word, and in that case its \emph{reset threshold} $\rt(A)$ is
the minimum length of a reset word.  If
$Q\mathbin{\cdot}w=\{s\}$, then $w$ \emph{resets $A$ to $s$}, and $s$ is the
\emph{reset state} of $w$.

The \emph{support} of $A$ is the directed multigraph with vertex set $Q$ and
an edge $q\to q\mathbin{\cdot}g$ for every $q\in Q$ and $g\in\Sigma$;
loops are allowed, and parallel edges carrying different input letters are
retained.  The automaton is \emph{strongly connected} if its support is strongly
connected, equivalently, if for every ordered pair $p,q\in Q$ there is a
word $w$ with $p\mathbin{\cdot}w=q$.

\begin{definition}\label{def:one-cluster-level}
For $a\in\Sigma$, its \emph{functional digraph}, also called the
\emph{$a$-skeleton}~\cite[Section~1 and Figure~1]{Steinberg2011PrimeCycle},
has vertex set $Q$ and the single outgoing edge
$q\longrightarrow q\mathbin{\cdot}a$ at each state.  The weakly connected
components, meaning the connected components of its underlying undirected
graph, are its \emph{$a$-clusters}~\cite[Section~3]{BealPerrin2009}.  We call $a$ a
\emph{one-cluster letter} if there is exactly one $a$-cluster, and then call
$A$ \emph{one-cluster with respect to $a$}.  The functional digraph has a
unique directed cycle, called the \emph{$a$-cycle}; let $C\subseteq Q$ be
its vertex set.  The elements of $C$ are \emph{cycle states}, and those of
$Q\setminus C$ are \emph{off-cycle states}.

For $q\in Q$, define its level relative to $a$ by
\[
  \ell_a(q)=\min\set{j\ge0:q\mathbin{\cdot}a^j\in C}.
\]
The \emph{level of $A$ with respect to $a$} is
\begin{equation}\label{eq:equivalent-level-definitions}
  \ell=\max_{q\in Q}\ell_a(q)
  =\min\set{j\ge0:Q\mathbin{\cdot}a^j\subseteq C}
  =\min\set{j\ge0:Q\mathbin{\cdot}a^j=C}.
\end{equation}
Finally, write
\[
  n=\abs Q,
  \qquad m=\abs C,
  \qquad t=n-m.
\]
\end{definition}

The last two minima in
\eqref{eq:equivalent-level-definitions} agree because $a$ restricts to a
permutation of $C$, so $C\subseteq Q\mathbin{\cdot}a^j$ for every $j$.
The one-cluster condition is also equivalent to Volkov's formulation that
some state can be reached from every state along a path labelled only by
$a$~\cite[Section~3.4]{Volkov2022Survey}.

If $\ell=0$, then $Q=C$ and $a$ acts on $Q$ as a cyclic permutation; in
this case the automaton is called \emph{circular} (with respect to $a$) in
the one-cluster literature~\cite[Section~1]{Steinberg2011PrimeCycle}; Volkov
uses the synonymous term \emph{cyclic}
~\cite[Section~3.4]{Volkov2022Survey}.
If $\ell\ge1$, minimality gives a state $q$ such that
$q\mathbin{\cdot}a^{\ell-1}\notin C$.  The states
$q,q\mathbin{\cdot}a,\ldots,q\mathbin{\cdot}a^{\ell-1}$ are distinct and
off-cycle: a repetition before entry into $C$ would create a directed cycle
different from the $a$-cycle.
Consequently,
\begin{equation}\label{eq:level-le-t}
  \ell\le t.
\end{equation}
When $\ell=0$, the same inequality holds because $t\ge0$.  Thus
\eqref{eq:level-le-t} is valid at every level.

We now formalize the relative-extension terminology used in the statement.
This is the cycle-restricted version of the standard method of extension;
the restriction to $C$ is essential at positive level.

\begin{definition}\label{def:relative-extension}
Let $\varnothing\ne S\subsetneq C$.  A word $u\in\Sigma^*$ \emph{extends
$S$ relative to $C$} if
\[
  \abs{Su^{-1}\cap C}>\abs S.
\]
Equivalently, $u$ is $S$-augmenting in the terminology of B\'eal and
Perrin~\cite[Section~3]{BealPerrin2009}.  If $u=wa^\ell$, we refer
descriptively to $w$ as the prefix before
the prescribed suffix $a^\ell$; the word that extends $S$ is the entire word
$u$, not merely $w$.
\end{definition}

For comparison, if $\varnothing\ne S\subsetneq C$, then a word
$u\in\Sigma^*$ \emph{extends $S$ in the ordinary sense} if
$\abs{Su^{-1}}>\abs S$.  Relative extension implies ordinary extension
because $Su^{-1}\cap C\subseteq Su^{-1}$, but the converse need not hold
when $C\ne Q$.  When $C=Q$, the two notions coincide.  At positive level,
\cref{thm:short-extension} is not a $1$-extensibility statement in Volkov's
sense: it imposes the stronger cycle-relative inequality, while its bound
applies to the prefix $w$ and the whole extending word may have length
$n+\ell$.  Moreover, Volkov's $\alpha$-extensibility quantifies over proper
nonsingleton subsets of $Q$, whereas the relative statement here also
includes singleton subsets of $C$~\cite[Section~3.4]{Volkov2022Survey}.

The restriction of the quantified subsets to $C$ is substantive.  If
ordinary extension is required for every proper nonsingleton subset of $Q$,
Berlinkov's one-cluster family $B_n$ has an obstruction: the shortest word
extending $\{0,n-1\}$ is $a^{n-2}ba^{n-2}$, of length $2n-3$
~\cite[Theorem~1]{Berlinkov2011CarpiDAlessandro}; see also
\cite[Section~3.4 and Figure~21]{Volkov2022Survey}.  For the distinguished
one-cluster letter in this family, $C=\{0,1,\ldots,n-2\}$, so the obstructing
set is not contained in $C$.  Thus the family rules out every uniform
coefficient below $2$ for ordinary extensibility, but it does not conflict
with the cycle-relative theorem proved here.

The short relative-extension theorem, \cref{thm:short-extension}, is the
local engine: it increases the relative-preimage cardinality of a proper
subset of the cycle using a word of the prescribed form and gives a uniform
bound on the prefix.  Iterating that engine, after an economical first
extension from a singleton, gives the reset bound in
\cref{thm:reset-bound}.

\begin{theorem}[Short relative extension]\label{thm:short-extension}
Let $A=\langle Q,\Sigma,\delta\rangle$ be a synchronizing one-cluster
automaton with respect to $a\in\Sigma$.  Let $C$ be its $a$-cycle, put
$n=\abs Q$ and $m=\abs C$, and let $\ell$ be its level with respect to $a$.
Assume $m\ge2$ and $\ell\ge1$.  For every nonempty proper subset
$S\subset C$, there exists a word $w\in\Sigma^*$ such that
\begin{equation}\label{eq:kks}
  \abs w\le n,
  \qquad
  \abs{S(wa^\ell)^{-1}\cap C}>\abs S.
\end{equation}
\end{theorem}

Thus $wa^\ell$ is a relative extending word in the sense of
\cref{def:relative-extension}, and the bound in \eqref{eq:kks} applies only
to the prefix $w$.  In particular, the whole relative extending word has
length at most $n+\ell$.

\begin{theorem}[Reset threshold]\label{thm:reset-bound}
Let $A=\langle Q,\Sigma,\delta\rangle$ be a synchronizing one-cluster
automaton with respect to $a\in\Sigma$.  Let $C$ be its $a$-cycle, put
$n=\abs Q$ and $m=\abs C$, and let $\ell$ be its level with respect to $a$.
Then
\begin{equation}\label{eq:reset-bound}
  \rt(A)\le (m-1)(n-1)+m\ell\le(n-1)^2.
\end{equation}
Consequently, every synchronizing one-cluster automaton satisfies the
\v{C}ern\'y conjecture.
\end{theorem}

The first inequality in \eqref{eq:reset-bound} is attained at positive level
in every order $n\ge4$; the examples and their exact reset thresholds are
given in \cref{sec:sharpness}.

The proof of \cref{thm:reset-bound} separates two boundary cases from the
positive-level argument.  If $m=1$, the power $a^\ell$ is a reset word.  If
$\ell=0$, the required estimate is Dubuc's circular-automaton theorem
~\cite[Proposition~4.6]{Dubuc1998}.  The formal proof in
\cref{sec:reset-proof} records both reductions.

\begin{example}[Running example]\label{ex:running}
Let $Q=\{0,1,2,3,4\}$, let $\Sigma=\{a,b\}$, and define the two letters by
\[
\begin{array}{c|ccccc}
q                    &0&1&2&3&4\\ \hline
q\mathbin{\cdot}a    &1&2&0&0&3\\
q\mathbin{\cdot}b    &0&3&1&4&1
\end{array}.
\]
The two letter actions are displayed in separate functional digraphs below.

\begin{center}
\begin{minipage}[t]{0.48\linewidth}
\centering
\vspace{0pt}
\begin{tikzpicture}
  \node[oc cycle state] (a0) at (0,0)          {$0$};
  \node[oc cycle state] (a1) at (1.45,0)       {$1$};
  \node[oc cycle state] (a2) at (0.725,-1.15)  {$2$};
  \node[oc state]       (a3) at (-1.45,0)      {$3$};
  \node[oc state]       (a4) at (-2.90,0)      {$4$};

  \draw[oc transition] (a0) -- (a1);
  \draw[oc transition] (a1) -- (a2);
  \draw[oc transition] (a2) -- (a0);
  \draw[oc transition] (a3) -- (a0);
  \draw[oc transition] (a4) -- (a3);

  \node[font=\small] at (0.725,-0.38) {$C$};
  \node[font=\scriptsize] at (-1.45,0.55) {$\ell_a(3)=1$};
  \node[font=\scriptsize] at (-2.90,0.55) {$\ell_a(4)=2$};
\end{tikzpicture}

\smallskip
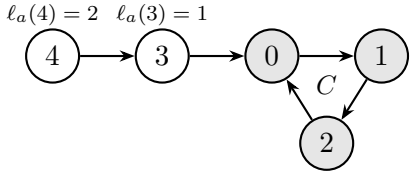
\small\textbf{Functional digraph of $a$.}
The shaded states form the $a$-cycle $C$.
\end{minipage}
\hfill
\begin{minipage}[t]{0.48\linewidth}
\centering
\vspace{0pt}
\begin{tikzpicture}
  \node[oc state] (b2) at (-1.45,0)      {$2$};
  \node[oc state] (b1) at (0,0)          {$1$};
  \node[oc state] (b3) at (1.45,0)       {$3$};
  \node[oc state] (b4) at (0.725,-1.15)  {$4$};
  \node[oc state] (b0) at (-1.45,-1.25)  {$0$};

  \draw[oc transition] (b2) -- (b1);
  \draw[oc transition] (b1) -- (b3);
  \draw[oc transition] (b3) -- (b4);
  \draw[oc transition] (b4) -- (b1);
  \draw[oc transition] (b0) edge[loop below,looseness=8] (b0);
\end{tikzpicture}

\smallskip
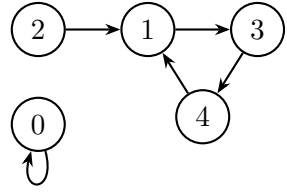
\small\textbf{Functional digraph of $b$.}
The state $0$ is fixed, and $1,3,4$ form a directed cycle.
\end{minipage}
\end{center}

The unique $a$-cycle is $C=\{0,1,2\}$, so
\[
  n=5,\qquad m=3,\qquad t=2,\qquad \ell=2.
\]
The automaton is strongly connected, meaning that for every ordered pair
$p,q\in Q$ some word maps $p$ to $q$: every state reaches $0$ by a power
of $a$, whereas $0$ reaches $1,2,3,4$ by the words $a,a^2,ab,abb$,
  respectively.  We shall return to this example after each main step.  In
particular, at the end we will exhibit a reset word of length
\[
  14=(m-1)(n-1)+m\ell.
\]
All calculations in this running example are illustrative and are not used
in the proof.
\end{example}

\section{Fitting and spectral decomposition}\label{sec:fitting}

Before filtering words by their controls, we isolate the structure forced by
the distinguished letter $a$ alone.  The Fitting decomposition produces two
complementary invariant subspaces: the restriction of the distinguished
transition is nilpotent on one and invertible of finite order on the other.
The associated polynomial projectors will be used both in filtered
generation and in the spectral-projector expansion.

This decomposition is the positive-level analogue of the polynomial
kernel--image splitting used by Dubuc for circular automata
~\cite[Lemma~3.1]{Dubuc1998}.  In the circular case only factors of
$X^m-1$ occur.  At positive level the identity
$T^\ell(T^m-I)=0$ introduces in addition the zero-primary, nilpotent
component isolated below.

In the rest of the proof, fix a one-cluster automaton
$A=\langle Q,\Sigma,\delta\rangle$ with respect to $a\in\Sigma$, and retain
the notation $C,n,m,t,\ell$ from \cref{def:one-cluster-level}.  Let
\[
  V=\C^{1\times n},
  \qquad T=P_a.
\]
We regard $T$ and all transition matrices as operators acting on $V$ by
right multiplication.
For $q\in Q$, let $e_q\in V$ denote the standard basis row whose $q$th
coordinate is $1$ and whose other coordinates are $0$.
Unless explicitly stated otherwise, every vector space, linear span, and
dimension in the paper is taken over $\C$.
Since every state is on $C$ after $\ell$ applications of $a$, and $a^m$
is the identity on $C$, row by row we have
\begin{equation}\label{eq:polynomial-identity}
  T^{\ell+m}=T^\ell.
\end{equation}

\Needspace{7\baselineskip}
Lemma~\ref{lem:fitting} makes this separation precise and supplies the dimension
data used in the later spectral and length estimates.

\begin{lemma}[Fitting decomposition]\label{lem:fitting}
Set
\[
  N=\set{v\in V:vT^\ell=0},
  \qquad R=VT^\ell.
\]
Then
\begin{equation}\label{eq:fitting}
  V=N\oplus R,
  \qquad \dim N=t,
  \qquad \dim R=m.
\end{equation}
Moreover, if $J=T|_N$, then $J^\ell=0$.  Relative to the basis
$(e_c)_{c\in C}$, the restriction $T|_R$ is the permutation operator induced
by the cyclic permutation $c\mapsto c\mathbin{\cdot}a$.  In particular,
$(T|_R)^m=I_R$.
\end{lemma}

\begin{proof}
Every row of $T^\ell$ is a standard basis row $e_c$ with $c\in C$, and
every such row occurs because $C$ is invariant and $a$ permutes it.  Thus
\[
  R=\Span_{\C}\set{e_c:c\in C},
  \qquad \dim R=m,
\]
and $T$ restricts to the cyclic permutation of those basis rows.  In
particular, every power of $T|_R$, including $T^\ell|_R$, is invertible.
Hence $N\cap R=0$.  The image of $T^\ell$ has dimension $m$, so
rank--nullity gives $\dim N=n-m=t$.  These dimensions and $N\cap R=0$ give
$V=N\oplus R$.  The definition of $N$ gives $J^\ell=0$.  Since $T|_R$ is
the permutation operator of an $m$-cycle, its $m$th power is $I_R$.
\end{proof}

We call $N$ the \emph{nilpotent Fitting summand} and $R$ the
\emph{invertible Fitting summand}.  These names refer only to the operator
identities $T^\ell|_N=0$ and $(T|_R)^m=I_R$.  In particular, they do not
refer to a probabilistic classification of states: the proof gives
$R=\Span_{\C}\{e_c:c\in C\}$, whereas $N$ need not be spanned by standard
basis rows.  Thus $R$ is concretely the cycle-coordinate row space; it is
not a set of ``recurrent states'' in Markov-chain terminology.

Steinberg's prime-cycle proof works over $\Q$ and uses the irreducibility of
the prime cyclotomic polynomial to keep all nontrivial eigenspaces of the
cyclic permutation together
~\cite[Lemma~4 and Theorem~8]{Steinberg2011PrimeCycle}.  We instead extend
scalars to $\C$ and split the invertible summand into its individual
eigenspaces.  This refinement is available for every $m$, including
composite $m$, and the nilpotent summand $N$ remains present when $\ell>0$.

Let $\muRoots_m$ denote the set of complex $m$th roots of unity.  For
$\lambda\in\muRoots_m$, define the right eigenspace
\[
  R_\lambda=\set{v\in R:vT=\lambda v}.
\]
Fix $c_0\in C$ and put
\[
  r_\lambda=\sum_{j=0}^{m-1}\lambda^{-j}
                    e_{c_0\mathbin{\cdot}a^j}.
\]
Then $r_\lambda T=\lambda r_\lambda$.  The coefficient matrix of the rows
$(r_\lambda)_{\lambda\in\muRoots_m}$ is a Vandermonde matrix on the $m$
distinct roots of $X^m-1$, and is therefore nonsingular.  Consequently,
\[
  R=\bigoplus_{\lambda\in\muRoots_m}R_\lambda,
  \qquad R_\lambda=\C r_\lambda,
  \qquad \dim R_\lambda=1.
\]
Every $v\in V$ thus has a unique decomposition
\[
  v=v_N+\sum_{\lambda\in\muRoots_m}v_\lambda,
  \qquad v_N\in N,\quad v_\lambda\in R_\lambda.
\]
Define the linear operators $E_N$ and $E_\lambda$ by
\[
  vE_N=v_N,
  \qquad
  vE_\lambda=v_\lambda.
\]
Thus $E_N$ projects onto $N$ along $R$, while $E_\lambda$ projects onto
$R_\lambda$ along all the other displayed direct-sum components.  We also
put
\[
  V_\lambda=VE_\lambda=R_\lambda
  \qquad(\lambda\in\muRoots_m).
\]
We also write
\begin{equation}\label{eq:E-R}
  E_R=I-E_N=\sum_{\lambda\in\muRoots_m}E_\lambda.
\end{equation}
We call $E_N$ the \emph{zero-primary projector} and the $E_\lambda$ the
\emph{nonzero eigenprojections}; collectively they are the spectral
projectors used below.  The terminology records operator components, not
classes of states.

\Needspace{7\baselineskip}
To use these projections inside spaces generated by transition matrices, we
need them to preserve every $T$-invariant subspace.  Lemma~\ref{lem:projectors} guarantees
this by realizing each projection as a polynomial in $T$; its resolution of
the identity will also start the later spectral-projector expansion.

\begin{lemma}[Polynomial projectors]\label{lem:projectors}
Each of $E_N$ and $E_\lambda$ $(\lambda\in\muRoots_m)$ is a polynomial in
$T$.  Furthermore,
\begin{equation}\label{eq:projector-sum}
  I=E_N+\sum_{\lambda\in\muRoots_m}E_\lambda.
\end{equation}
\end{lemma}

\begin{proof}
If $\ell=0$, then $N=0$ and $E_N=0$.  In that case the minimal polynomial
of $T$ divides $X^m-1$, whose roots are distinct over $\C$; ordinary
Lagrange interpolation gives each $E_\lambda$ as a polynomial in $T$, and
their sum is $I$.  We may therefore assume $\ell\ge1$.

The minimal polynomial of $T$ divides $X^\ell(X^m-1)$.  The factors
$X^\ell$ and $X-\lambda$ for $\lambda\in\muRoots_m$ are pairwise coprime.
The Chinese remainder theorem supplies polynomials $e_N$ and $e_\lambda$
with the following residue classes:
\[
\begin{array}{c|cc}
 & X^\ell & X-\mu\quad(\mu\in\muRoots_m)\\ \hline
e_N       &1&0\\
e_\lambda &0&\begin{cases}1,&\mu=\lambda,\\0,&\mu\ne\lambda.
                 \end{cases}
\end{array}.
\]
On $N$, the operator $T$ is annihilated by $X^\ell$; on $R_\mu$, it acts
as multiplication by $\mu$.  Therefore $e_N(T)$ is the identity on $N$ and
zero on every $R_\mu$, while $e_\lambda(T)$ is the identity on $R_\lambda$
and zero on every other direct-sum component.  Thus these evaluated
polynomials are precisely $E_N$ and $E_\lambda$.  Their sum acts as the
identity on every component of $V=N\oplus\bigoplus_\mu R_\mu$, which proves
\eqref{eq:projector-sum}.
\end{proof}

Define the zero-sum subspace
\begin{equation}\label{eq:row-zero-sum}
  \calZ=\set{v\in V:v\one=0}.
\end{equation}
Every transition matrix preserves $\calZ$ on the right.  The $1$-eigenspace
of $T$ in $R$ is spanned by $[C]$, and $[C]\notin\calZ$.  We therefore
have
\begin{equation}\label{eq:R-zero-sum}
  R\cap\calZ
   =\bigoplus_{\lambda\in\muRoots_m\setminus\{1\}}V_\lambda,
  \qquad \dim(R\cap\calZ)=m-1.
\end{equation}
Indeed, if $v\in V_\lambda$ and $\lambda\ne1$, then
$v\one=vT\one=\lambda v\one$, so $v\one=0$.
Notice also that $N\subseteq\calZ$: if $vT^\ell=0$, then
$v\one=vT^\ell\one=0$.

\paragraph{The running example: Fitting data.}
Let $e_i$ denote the standard row at state $i$.  For
\cref{ex:running},
\[
  R=\Span\set{e_0,e_1,e_2},
  \qquad
  N=\Span\set{\xi_1,\xi_2},
  \qquad
  \xi_1=e_3-e_2,\qquad \xi_2=e_4-e_1.
\]
Directly from the $a$-arrows,
\[
  \xi_2T=\xi_1,\qquad \xi_1T=0.
\]
Thus $J^2=0$ and $\rank J=1$, while $T|_R$ is the permutation operator of a
$3$-cycle.  In this example the minimal polynomial is
$X^2(X^3-1)$, and the two Fitting projectors are especially simple:
$E_R=T^3$ and $E_N=I-T^3$.

\section{Filtered spaces and short actual-word generators}\label{sec:filtered}

When $\ell=0$, so that $C=Q$, the rows and filtered spaces below coincide
over $\C$ with Dubuc's construction for circular automata
~\cite[Definition~4.2]{Dubuc1998}.  At positive level, the auxiliary
filtered spaces defined below separate the nilpotent part from the
cycle-coordinate part.  This refinement has no nonzero nilpotent counterpart
in the circular case.

After fixing a letter $f$, we use a row that records the change of the
cycle-indicator row after applying $f$.  We organize its right translates
by the number of controls in the translating word.  The aim is to represent
the resulting filtered spaces by actual words with explicit length bounds.

Every letter in $\Sigma\setminus\{a\}$ is called a \emph{control letter}, or
briefly a \emph{control}.  For a word $w$, write
$\abs{w}_{\Sigma\setminus\{a\}}$ for the number of
occurrences in $w$ of letters different from $a$; equivalently, this is the
number of control letters in $w$, or the length remaining after every
occurrence of $a$ is deleted.  Fix $f\in\Sigma$ and define
\begin{equation}\label{eq:uf}
  u_f=[C](P_f-I)\in\calZ.
\end{equation}
For $r\ge0$, let
\begin{equation}\label{eq:filtered-spaces}
  \calL_r(f)=\Span\set{u_fP_w:
               \abs{w}_{\Sigma\setminus\{a\}}\le r},
  \qquad
  \calM_r(f)=\calL_r(f)T^\ell.
\end{equation}
We usually suppress $f$.  Put
\[
  \calN_r=\calL_r\cap N,
  \qquad k_r=\dim\calN_r,
  \qquad d_r=\dim\calM_r.
\]
Each $\calL_r$ is invariant under right multiplication by $T$.  Since the
projectors in \cref{lem:projectors} are polynomials in $T$, they preserve
$\calL_r$.  Consequently $\calN_r$ and $\calM_r$ are $T$-invariant.  The
definitions also give
\begin{equation}\label{eq:filtration-inclusions}
  \calL_{r-1}\subseteq\calL_r,
  \qquad
  \calN_{r-1}\subseteq\calN_r,
  \qquad
  \calM_{r-1}\subseteq\calM_r
  \qquad(r\ge1).
\end{equation}

\Needspace{7\baselineskip}
A vector in $\calL_r$ may have nonzero components in both Fitting summands
$N$ and $R$.  Lemma~\ref{lem:compatible-splitting} proves that the corresponding component spaces
are $\calN_r=\calL_r\cap N$ and $\calM_r=\calL_r\cap R$, and that this
direct-sum decomposition is compatible with the filtration.  Its dimension
bounds are the input for the generator-length estimates below.

\begin{lemma}[Compatible splitting]\label{lem:compatible-splitting}
For the fixed letter $f$ in \eqref{eq:uf} and every $r\ge0$,
\begin{equation}\label{eq:filtered-splitting}
  \calM_r=\calL_r\cap R,
  \qquad
  \calL_r=\calN_r\oplus\calM_r.
\end{equation}
In particular,
\begin{equation}\label{eq:dimension-bounds}
  k_r\le t,
  \qquad d_r\le m-1.
\end{equation}
\end{lemma}

\begin{proof}
Because $\calL_r$ is $T$-invariant, $\calL_rT^\ell\subseteq\calL_r$.
Because $VT^\ell=R$, the same rows lie in $R$.  Hence
$\calL_rT^\ell\subseteq\calL_r\cap R$.  Conversely,
if $v\in\calL_r\cap R$, choose $p\in\{0,\ldots,m-1\}$ with
$p+\ell\equiv0\pmod m$.  Since $T^m|_R=I_R$, we have
$T^{p+\ell}|_R=I_R$.  Then
$v=(vT^p)T^\ell\in\calL_rT^\ell$.  Applying the polynomial projections
$E_N$ and $E_R=I-E_N$ to any $v\in\calL_r$ gives
$vE_N\in\calL_r\cap N=\calN_r$ and
$vE_R\in\calL_r\cap R=\calM_r$.  Their intersection is zero by
$V=N\oplus R$, which proves the direct sum.  Finally,
$\calL_r\subseteq\calZ$, so \eqref{eq:dimension-bounds} follows from
\cref{lem:fitting,eq:R-zero-sum}.
\end{proof}

Relative to the fixed letter $f$ in \eqref{eq:uf}, an
\emph{actual-word generating family} for a subspace $W\subseteq V$ is a
family of labelled rows indexed by a finite set $\Gamma$:
\[
  \bigl(w_\gamma,x_\gamma\bigr)_{\gamma\in\Gamma},
  \qquad
  w_\gamma\in\Sigma^*,\quad x_\gamma=u_fP_{w_\gamma},\quad
  W=\Span\set{x_\gamma:\gamma\in\Gamma}.
\]
Each pair $(w_\gamma,x_\gamma)$ is an \emph{actual-word generator}.  Linear
combinations use the row $x_\gamma$, whereas length, number of controls, and
terminal-suffix conditions refer to its chosen representing word
$w_\gamma$.  Keeping this label is necessary because distinct words may
represent the same row vector.
For a row vector $v$, the notation $v\C[T]$ means
$\Span\set{vq(T):q\in\C[X]}$, so all cyclic modules below are right
$\C[T]$-modules.  We write
\[
  \operatorname{Ann}_T(v)=
  \set{q(X)\in\C[X]:vq(T)=0}
\]
for the annihilator ideal of $v$.  Empty spaces have the empty generating
family, and the minimal polynomial of the zero space is understood to be
$1$.

\Needspace{7\baselineskip}
Lemma~\ref{lem:generation} gives two complementary length bounds for actual-word
generators of $\calL_r$.  The first bound contains the total dimension
$d_r+k_r$, whereas the second adds at most $\ell$ for the nilpotent part
introduced at each filtration step.  Neither estimate uniformly dominates
the other, so we retain both constructions.  After passage to
$\calM_r=\calL_r\cap R$, they yield competing terminal-$a^\ell$ bounds whose
minimum is used in the annular induction, while the total-dimension estimate
will later force the filtration $(\calM_r)$ to stabilize.

\begin{lemma}[Two filtered generation bounds]\label{lem:generation}
For the fixed letter $f$ in \eqref{eq:uf} and every $r\ge0$, the space
$\calL_r(f)$ has two possibly empty actual-word
generating families.  Every representing word $w_\gamma$ in the first and
second family satisfies, respectively,
\begin{align}
  \abs{w_\gamma}&\le d_r+k_r-1, \label{eq:gen-total}\\
  \abs{w_\gamma}&\le d_r+\ell r+\ell-1. \label{eq:gen-jordan}
\end{align}
Every representing word in either family has at most $r$ letters different
from $a$.
Moreover, $\calM_r$ has an actual-word generating family in which every
representing word has at most $r$ letters different from $a$, ends in
$a^\ell$, and satisfies
\begin{equation}\label{eq:gen-M}
  \abs{w_\gamma}\le d_r+\ell-1+\min(k_r,\ell r).
\end{equation}
\end{lemma}

\begin{proof}
If $u_f=0$, all the filtered spaces are zero and the empty families prove
the assertions.  Hence assume $u_f\ne0$.

\smallskip
\noindent\emph{The total-dimension bound.}
We first prove \eqref{eq:gen-total}.  The space $\calL_0=u_f\C[T]$ is the
cyclic right $\C[T]$-module generated by $u_f$, so the words
$\varepsilon,a,\ldots,a^{d_0+k_0-1}$ represent a generating family.  Suppose
$r\ge1$ and that an actual-word generating family
$((w_\gamma,x_\gamma))_{\gamma\in\Gamma}$ has been chosen for
$\calL_{r-1}$.  Modulo $\calL_{r-1}$, the space
$\calL_r$ is generated as a right $\C[T]$-module by the sources
\begin{equation}\label{eq:filtration-sources}
  x_\gamma P_g,
  \qquad \gamma\in\Gamma,\quad g\in\Sigma\setminus\{a\}.
\end{equation}
We call these vectors the \emph{sources at filtration index $r$} relative
to the chosen family; their representing words are $w_\gamma g$.
Indeed, modulo $\calL_{r-1}$ every relevant word has exactly $r$ controls
and can be written $vga^j$, where
$\abs{v}_{\Sigma\setminus\{a\}}=r-1$, $g\ne a$, and $j\ge0$.
Then $u_fP_{vga^j}=(u_fP_v)P_gT^j$, and $u_fP_v$ can be expanded in the
rows $x_\gamma$ of the chosen family.
Let
\[
  \Delta_r=\dim\calL_r-\dim\calL_{r-1}
          =(d_r-d_{r-1})+(k_r-k_{r-1}).
\]
If $\Delta_r=0$, the quotient is zero and no new representative is needed.
Otherwise, let $\overline T$ be the operator induced by $T$ on the
$\Delta_r$-dimensional quotient $\calL_r/\calL_{r-1}$.  By the
Cayley--Hamilton theorem, the $\overline T$-orbit of every source coset is
spanned by its first $\Delta_r$ powers.  Thus it suffices to append one
control and at most $\Delta_r-1$ copies of $a$; the maximum length rises by
at most $\Delta_r$.  The base maximum is $d_0+k_0-1$.  After filtration
indices $1,\ldots,r$, the maximum is at most
\[
  d_0+k_0-1+\sum_{i=1}^r\Delta_i=d_r+k_r-1.
\]
This proves \eqref{eq:gen-total}.

\smallskip
\noindent\emph{The spectral--nilpotent bound.}
For \eqref{eq:gen-jordan}, define the spectral support
\[
  \Lambda_r=\set{\lambda\in\muRoots_m\setminus\{1\}:
                       \calM_r\cap V_\lambda\ne0}.
\]
Because $\calM_r$ is $T$-invariant, \eqref{eq:R-zero-sum} and the
one-dimensionality of every $V_\lambda$ give
\[
  \calM_r=\bigoplus_{\lambda\in\Lambda_r}V_\lambda,
  \qquad \Lambda_{r-1}\subseteq\Lambda_r\quad(r\ge1).
\]
Thus the minimal polynomial $p_r(X)$ of $T|_{\calM_r}$ is
$\prod_{\lambda\in\Lambda_r}(X-\lambda)$.  If $\calM_r=0$, this product is
empty and $p_r=1$.  Consequently,
\[
  \deg p_r=d_r,
  \qquad p_{r-1}\mid p_r\quad(r\ge1).
\]
For $r\ge1$, set
\[
  \psi_r(X)=\frac{p_r(X)}{p_{r-1}(X)},
  \qquad s_r=\deg \psi_r=d_r-d_{r-1}.
\]
The roots of $\psi_r$ are exactly
$\Lambda_r\setminus\Lambda_{r-1}$.  Define
\[
  \Theta_r:\calM_r/\calM_{r-1}
  \longrightarrow
  \bigoplus_{\lambda\in\Lambda_r\setminus\Lambda_{r-1}}V_\lambda,
  \qquad
  \Theta_r(v+\calM_{r-1})
  =\sum_{\lambda\in\Lambda_r\setminus\Lambda_{r-1}}vE_\lambda.
\]
An element of $\calM_{r-1}$ has zero component in every new eigenline, so
$\Theta_r$ is well defined.  If $\Theta_r(v+\calM_{r-1})=0$, then $v$ has
only old eigenline components and therefore belongs to $\calM_{r-1}$; the
coset is zero.  Thus $\Theta_r$ is injective.  Every new eigenline lies in
its image, so it is also surjective.  Hence $\Theta_r$ is an isomorphism
of vector spaces.  Moreover, every $E_\lambda$ commutes with $T$, so
\[
  \Theta_r\bigl((v+\calM_{r-1})T\bigr)
  =\Theta_r(v+\calM_{r-1})T.
\]
By iteration and linearity the same identity holds with $T$ replaced by
every polynomial in $T$.  Hence $\Theta_r$ is an isomorphism of right
$\C[T]$-modules, and
\begin{equation}\label{eq:spectral-quotient}
  \calM_r/\calM_{r-1}
  \cong\bigoplus_{\lambda\in\Lambda_r\setminus\Lambda_{r-1}}V_\lambda,
\end{equation}
and this quotient has minimal polynomial $\psi_r$ of degree $s_r$.
For any $y\in\calL_{r-1}$ and $g\in\Sigma\setminus\{a\}$, put
$x=yP_g$.  Its $R$-component belongs to $\calM_r$; multiplying by
$\psi_r(T)$ kills precisely
its components in those new eigenlines, while every remaining nonzero
component lies in $\calM_{r-1}$.  Its $N$-component is killed by the
subsequent $T^\ell$.  Hence
\begin{equation}\label{eq:annihilating-relation}
  x\psi_r(T)T^\ell=(xT^\ell)\psi_r(T)
  \in\calM_{r-1}\subseteq\calL_{r-1}.
\end{equation}
The polynomial $X^\ell\psi_r(X)$ is monic of degree $s_r+\ell$, so, modulo
$\calL_{r-1}$, the orbit of each source is generated by
\[
  x,xT,\ldots,xT^{s_r+\ell-1}.
\]
All sources share the same polynomial; their number affects the number of
generators, not their maximum length.

For $r=0$, the projector splitting gives
\[
  \calL_0=u_f\C[T],\qquad
  \calN_0=u_fE_N\C[T],\qquad
  \calM_0=u_fE_R\C[T]=u_fT^\ell\C[T].
\]
The last equality holds because $T^\ell$ vanishes on $N$ and is invertible
on $R$.  Put $x=u_fE_N$.  If $x=0$, set $h=0$; then
$x\C[T]=0$ and $\operatorname{Ann}_T(x)=\C[X]=(X^h)$.  Suppose now that
$x\ne0$, and let $h\ge1$ be the least integer for which $xT^h=0$.
The rows
\[
  x,xT,\ldots,xT^{h-1}
\]
are linearly independent.  Indeed, in a nonzero relation choose the least
index $k$ with nonzero coefficient.  The relation has the form
$xT^k(\beta I+Tq(T))=0$, where $\beta\ne0$.  The restriction of $T$ to $N$
is nilpotent, so $\beta I+Tq(T)$ is invertible on $N$.  It would follow that
$xT^k=0$, contrary to $k<h$.  Hence the map
\[
  \C[X]/(X^h)\longrightarrow x\C[T],
  \qquad q+(X^h)\longmapsto xq(T),
\]
is an isomorphism.  In both cases
$\operatorname{Ann}_T(x)=(X^h)$ and
$h=\dim(x\C[T])=k_0\le\ell$.  The vector
$u_fE_R$ cyclically generates $\calM_0$, so
$\operatorname{Ann}_T(u_fE_R)$ has monic generator $p_0$.
Since $p_0(0)\ne0$,
$X^hp_0(X)$ annihilates $u_f$.  Hence the labelled pairs
\[
  (a^j,u_fT^j),\qquad 0\le j<h+d_0,
\]
form an actual-word generating family for $\calL_0$.  Their maximum
representing-word length is at most
$d_0+\ell-1$, which is \eqref{eq:gen-jordan} for $r=0$.  For the induction
step, suppose that
$\mathcal G_{r-1}=((w_\gamma,x_\gamma))_{\gamma\in\Gamma}$ satisfies
\eqref{eq:gen-jordan} for $\calL_{r-1}$.  Retain those labelled pairs and,
for every $\gamma\in\Gamma$, $g\in\Sigma\setminus\{a\}$, and
$0\le j<s_r+\ell$, adjoin
\[
  \bigl(w_\gamma ga^j,\ x_\gamma P_gT^j\bigr).
\]
The source description in \eqref{eq:filtration-sources} and
\eqref{eq:annihilating-relation} show that the
resulting rows span $\calL_r$.  The maximum representing-word length
increases by at most $s_r+\ell$.  More explicitly, it is at most
\[
  \bigl(d_{r-1}+\ell(r-1)+\ell-1\bigr)+(s_r+\ell)
  =d_r+\ell r+\ell-1,
\]
because $s_r=d_r-d_{r-1}$.  This proves \eqref{eq:gen-jordan}.  When
$s_r=0$, the new rows
use only $j=0,\ldots,\ell-1$, as required.

\smallskip
\noindent\emph{Generating the invertible component.}
It remains to generate $\calM_r$ directly.  One complete generating family
is obtained from \eqref{eq:gen-total} by replacing each labelled pair
$(w_\gamma,x_\gamma)$ with
$(w_\gamma a^\ell,x_\gamma T^\ell)$; its maximum length is
\begin{equation}\label{eq:M-route-one}
  d_r+k_r+\ell-1.
\end{equation}
For the other family, argue by induction on $r$.  At $r=0$, use
the labelled pairs
\[
  (a^{\ell+j},u_fT^{\ell+j}),\qquad 0\le j<d_0.
\]
If $d_0=0$, this list is empty and generates the zero space $\calM_0$.
Now assume $r\ge1$, retain the family
already constructed for $\calM_{r-1}$, let
$\mathcal G_{r-1}=((w_\gamma,v_\gamma))_{\gamma\in\Gamma}$ be the family
for $\calL_{r-1}$ furnished by \eqref{eq:gen-jordan}, and set
\[
  \mathcal S_r=\set{v_\gamma P_g:\gamma\in\Gamma,\
                       g\in\Sigma\setminus\{a\}}.
\]
Applying $T^\ell$ to the source description in \eqref{eq:filtration-sources}
and passing to the
quotient gives
\[
  \calM_r/\calM_{r-1}
  =\sum_{x\in\mathcal S_r}
       (xT^\ell+\calM_{r-1})\C[T].
\]
More explicitly,
\begin{equation}\label{eq:M-quotient-generators}
  \calM_r=\calM_{r-1}+
  \Span\set{xT^{\ell+j}:x\in\mathcal S_r,\ 0\le j<s_r}.
\end{equation}
Indeed, by \eqref{eq:spectral-quotient}, division by the monic polynomial
$\psi_r$ reduces
every power of $T$ to a linear combination of $I,T,\ldots,T^{s_r-1}$,
which proves \eqref{eq:M-quotient-generators}.  The new rows represented by
the labelled pairs
\[
  \bigl(w_\gamma ga^{\ell+j},v_\gamma P_gT^{\ell+j}\bigr),
  \qquad \gamma\in\Gamma,\quad g\in\Sigma\setminus\{a\},
  \quad 0\le j<s_r,
\]
have maximum representing-word length at most
\[
  \bigl(d_{r-1}+\ell(r-1)+\ell-1\bigr)
  +(1+\ell+s_r-1)
  =d_r+\ell r+\ell-1.
\]
The representing words inherited from $\calM_{r-1}$ are shorter by the
induction hypothesis.  If $s_r=0$, the new list in
\eqref{eq:M-quotient-generators} is empty and
$\calM_r=\calM_{r-1}$.  Every construction at filtration index $r$ appends
at most one new control letter to a word from filtration index $r-1$ and
otherwise appends only powers of $a$.  Thus all stated families have at
most $r$ controls.
Taking the shorter of the two complete families, \eqref{eq:M-route-one}
and this second family, proves \eqref{eq:gen-M}.
\end{proof}

\paragraph{The running example: filtered growth.}
Write $U=P_b$.  Using the vectors $\xi_1,\xi_2$ from the Fitting-data
calculation for \cref{ex:running}, we obtain
\[
  u_b=[C](U-I)=e_3-e_2=\xi_1,
  \qquad
  u_bU=e_4-e_1=\xi_2.
\]
Because $u_bT=0$, a word with at most one control can produce only
$\xi_1$ and $\xi_2$; hence
\[
  \calL_0=\Span\set{\xi_1},\qquad
  \calL_1=N,\qquad
  \calM_0=\calM_1=0.
\]
Thus allowing one control enlarges $\calL_0=\C\xi_1$ to the whole
two-dimensional space $N$, while its $R$-component remains zero.  A second
control produces a nonzero $R$-component: since
\[
  u_bU^2T^2=e_0-e_1,
  \qquad
  u_bU^2T^3=e_1-e_2,
\]
these two rows span $R\cap\calZ$.  Consequently,
\[
\begin{array}{c|c|c|c}
r&\calN_r&\calM_r&(k_r,d_r)\\ \hline
0&\Span\set{\xi_1}&0&(1,0)\\
1&N&0&(2,0)\\
r\ge2&N&R\cap\calZ&(2,2).
\end{array}.
\]
In particular, $\calL_2=\calZ$.  The actual words $\mathtt{bbaa}$ and
$\mathtt{bbaaa}$ represent rows that generate $\calM_2$; the maximum
representing-word length is $5$, exactly the bound in \eqref{eq:gen-M}.
Likewise, the rows represented by
$\varepsilon,\mathtt b,\mathtt{bb},\mathtt{bba}$ generate $\calL_2$, with
maximum representing-word length $3=d_2+k_2-1$.

For later use, define
\begin{equation}\label{eq:kappa}
  \kappa_r=\min(t,\ell r).
\end{equation}
For every nonnegative integer tuple
$\mathbf h=(h_0,\ldots,h_r)\in\mathbb Z_{\ge0}^{r+1}$, define its
\emph{total exponent} by
\begin{equation}\label{eq:total-exponent}
  \tau(\mathbf h)=h_0+\cdots+h_r.
\end{equation}
By \cref{lem:generation,eq:dimension-bounds}, $\calM_r$ has an actual-word
generating family whose representing words end in $a^\ell$ and have length
at most
\begin{equation}\label{eq:coarse-M-bound}
  m+\ell-2+\kappa_r.
\end{equation}
If such a representing word has exactly $r$ letters different from $a$, it
has a unique block form
\begin{equation}\label{eq:block-form}
  a^{h_0}g_1a^{h_1}\cdots g_ra^{h_r+\ell},
  \qquad g_i\ne a,\quad h_i\ge0.
\end{equation}
For the exponent tuple $\mathbf h=(h_0,\ldots,h_r)$ in
\eqref{eq:block-form}, \eqref{eq:coarse-M-bound} gives
\begin{equation}\label{eq:generator-total-exponent-bound}
  \tau(\mathbf h)\le m+\kappa_r-r-2.
\end{equation}

\section{The annular descent identity}\label{sec:annular}

At each filtration index $r$, we sum the transition rows associated with
block words whose free $a$-exponents lie in a finite set defined below.  In
this section, \emph{descent from $r$ to $r-1$} means membership of that sum
in $\calM_{r-1}$.  The simultaneous prefix bound will allow
\cref{sec:short-extension-proof}
to apply the no-short-extension hypothesis term by term.

Throughout this section, assume $m\ge2$ and $\ell\ge1$.
For $1\le r\le n-1$, set
\begin{equation}\label{eq:annular-endpoints}
  D_r^- =\max(0,\kappa_r-r),
  \qquad
  D_r^+ =m+\kappa_r-r-1.
\end{equation}
Using the total exponent from \eqref{eq:total-exponent}, define the
\emph{annular exponent set}
\begin{equation}\label{eq:annular-exponent-set}
  \mathcal A_r=
  \set{\mathbf h\in\mathbb Z_{\ge0}^{r+1}:
             D_r^-\le\tau(\mathbf h)\le D_r^+}.
\end{equation}
We call $[D_r^-,D_r^+]\cap\mathbb Z$ its \emph{total-exponent interval}.
This interval is nonempty.  If $\kappa_r\ge r$, then
$D_r^+-D_r^-=m-1\ge1$.  If $\kappa_r<r$, the inequality
$\ell r\ge r$ forces $\kappa_r=t$; hence
$D_r^-=0$ and $D_r^+=n-r-1\ge0$.

For comparison, formally specializing these definitions to $\ell=0$ gives
$t=\kappa_r=0$, $D_r^-=0$, $D_r^+=n-r-1$, and
$\calM_r=\calL_r$.  The annular exponent set then becomes the simplex
$\{\mathbf h\in\mathbb Z_{\ge0}^{r+1}:\tau(\mathbf h)\le n-r-1\}$.
The corresponding descent relation is Dubuc's circular relation
~\cite[Lemma~4.4]{Dubuc1998}.  In this paper, the shifted endpoints and
descent modulo $\calM_{r-1}$ are used only under the positive-level
assumptions $m\ge2$ and $\ell\ge1$.

For $f\in\Sigma$ and $g_1,\ldots,g_r\in\Sigma\setminus\{a\}$, write
$G_i=P_{g_i}$.
Define
\begin{equation}\label{eq:omega}
  \Omega_r=
  \sum_{\mathbf h\in\mathcal A_r}
  u_fT^{h_0}G_1T^{h_1}\cdots G_rT^{h_r+\ell}.
\end{equation}
We call $\Omega_r$ the \emph{annular sum}; each displayed summand is an
\emph{annular term}.  Thus ``annular'' has the precise meaning given by
\eqref{eq:annular-exponent-set}.

\Needspace{7\baselineskip}
The required descent and prefix estimate are recorded in
\cref{prop:annular-descent}.

\begin{proposition}[Annular descent]\label{prop:annular-descent}
Let $A=\langle Q,\Sigma,\delta\rangle$ be one-cluster with respect to
$a\in\Sigma$.  Let $C$ be its $a$-cycle, put
$n=\abs Q$, $m=\abs C$, and $t=n-m$, and let $\ell$ be its level with
respect to $a$.  Assume $m\ge2$ and $\ell\ge1$, and use the matrices and
filtered spaces defined in
\cref{eq:uf,eq:filtered-spaces,eq:kappa,eq:annular-endpoints}.  For every
$1\le r\le n-1$, every
$f\in\Sigma$, and all
$g_1,\ldots,g_r\in\Sigma\setminus\{a\}$, the annular sum $\Omega_r$ in
\eqref{eq:omega} satisfies
\begin{equation}\label{eq:annular-membership}
  \Omega_r\in\calM_{r-1}(f).
\end{equation}
For every $\mathbf h\in\mathcal A_r$, the word
\[
  fa^{h_0}g_1a^{h_1}\cdots g_ra^{h_r},
\]
obtained by prepending $f$ to the indexing block word and removing its
terminal $a^\ell$ has length at most $n$.
\end{proposition}

\paragraph{The running example: annular cancellation.}
For \cref{ex:running}, $\kappa_1=\kappa_2=2$, so the total-exponent
intervals for filtration indices $1$ and $2$ are
\[
  [D_1^-,D_1^+]=[1,3],
  \qquad
  [D_2^-,D_2^+]=[0,2].
\]
Take $f=g_1=g_2=b$ and continue to write $U=P_b$.  At filtration index $1$ every
annular term is zero: $u_bT=0$, while $u_bU=\xi_2$ and
$\xi_2T^2=0$.  At filtration index $2$, the set $\mathcal A_2$ contains
ten exponent triples,
but only three give nonzero rows:
\[
\begin{array}{c|c}
(h_0,h_1,h_2)&u_bT^{h_0}UT^{h_1}UT^{h_2+2}\\ \hline
(0,0,0)&e_0-e_1\\
(0,0,1)&e_1-e_2\\
(0,0,2)&e_2-e_0.
\end{array}.
\]
They cancel by the vector identity
$(e_0-e_1)+(e_1-e_2)+(e_2-e_0)=0$, and hence
$\Omega_2=0\in\calM_1$.  The other seven exponent triples already
vanish in the nilpotent calculation.  Here $\calM_1=0$, so descent is exact
cancellation.  In general \cref{prop:annular-descent} asserts only the
membership $\Omega_r\in\calM_{r-1}$.

\Needspace{7\baselineskip}
Lemma~\ref{lem:nilpotent-growth} is the only dimension estimate needed for the
nilpotent restriction.  It states precisely how much the span of the
iterates $x_jJ^k$ grows when $x_jJ^{\nu_j}$ is not already contained in
the span generated at earlier indices.

\begin{lemma}[Dimension growth for a nilpotent operator]
\label{lem:nilpotent-growth}
Let $W$ be a finite-dimensional complex vector space, let
$J\in\operatorname{End}(W)$ satisfy $J^\ell=0$ with $\ell\ge1$, and let
$z\ge0$ and $x_1,\ldots,x_z\in W$.  We use right-action notation
$xJ^k:=J^k(x)$.  Choose integers $0\le\nu_j<\ell$ for $1\le j\le z$.  For
$0\le j\le z$, put
\[
  H_j=\Span\set{x_iJ^k:1\le i\le j,\ 0\le k<\ell},
  \qquad H_0=\{0\}.
\]
Then every $H_j$ is $J$-invariant.  For each $1\le j\le z$, either
\[
  x_jJ^{\nu_j}\in H_{j-1},
\]
or
\[
  \dim H_j-\dim H_{j-1}\ge\nu_j+1.
\]
\end{lemma}

\begin{proof}
The definition and $J^\ell=0$ show that every $H_j$ is $J$-invariant.
Fix $j$ and suppose $x_jJ^{\nu_j}\notin H_{j-1}$.  Work in the quotient
$W/H_{j-1}$, let $\overline J$ be the operator induced by $J$, and write
$y$ for the image of $x_j$.  Then $y\overline J^{\nu_j}\ne0$, and
$y,y\overline J,\ldots,y\overline J^{\nu_j}$ are linearly independent.
Indeed, in a nontrivial relation choose the least index $k$ having nonzero
coefficient.  There are $\beta\in\C\setminus\{0\}$ and $q\in\C[X]$ such
that the relation has the form
\[
  y\overline J^k
  \bigl(\beta I+\overline Jq(\overline J)\bigr)=0.
\]
Since $\overline J$ is nilpotent,
$\beta I+\overline Jq(\overline J)$ is invertible.  Thus
$y\overline J^k=0$, which would force
$y\overline J^{\nu_j}=0$.
The displayed $\nu_j+1$ independent quotient vectors lie in
$H_j/H_{j-1}$, which proves the dimension inequality.
\end{proof}

For a finite set $Y$ of complex numbers and $d\ge0$, define
\[
  h_d(Y)=
  \sum_{\substack{(k_\lambda)_{\lambda\in Y}\in\mathbb Z_{\ge0}^{Y}\\
                   \sum_{\lambda\in Y}k_\lambda=d}}
       \prod_{\lambda\in Y}\lambda^{k_\lambda}.
\]
Thus $h_d(Y)$ is the complete homogeneous symmetric polynomial of degree
$d$ in the elements of $Y$.  We also set $h_d(Y)=0$ for $d<0$ and use the
empty-product convention $h_0(Y)=1$.
For a formal power series $\varphi(x)$, the notation
$[x^d]\varphi(x)$ denotes the
coefficient of $x^d$.

\Needspace{7\baselineskip}
Lemma~\ref{lem:coefficient} evaluates the scalar sum over the exponents attached to
nonzero eigenvalue projectors.  It converts that sum into coefficients of a
roots-of-unity generating function, whose $m$-step recurrence gives the
required annular cancellation.

\begin{lemma}[Annular coefficient identity]\label{lem:coefficient}
Let $Y_0$ be a finite set and let $Y=Y_0\cup\{1\}$, where $1\notin Y_0$.
For all integers $L\le U$,
\begin{equation}\label{eq:partial-h}
  \sum_{k=L}^{U}h_k(Y_0)=h_U(Y)-h_{L-1}(Y).
\end{equation}
If $Y\subseteq\muRoots_m$, then
\begin{equation}\label{eq:roots-generating-function}
  \sum_{d\ge0}h_d(Y)x^d
  =\frac{\Phi_Y(x)}{1-x^m},
  \qquad
  \Phi_Y(x)=\prod_{\lambda\in\muRoots_m\setminus Y}(1-\lambda x).
\end{equation}
\end{lemma}

\begin{proof}
Adding a variable equal to $1$ gives
$h_d(Y)=\sum_{k=0}^dh_k(Y_0)$ for $d\ge0$; subtraction, together with the
negative-index convention, proves \eqref{eq:partial-h}.  Moreover,
\[
  \sum_{d\ge0}h_d(Y)x^d=\prod_{\lambda\in Y}(1-\lambda x)^{-1}.
\]
Since $\prod_{\lambda\in\muRoots_m}(1-\lambda x)=1-x^m$, this is
\eqref{eq:roots-generating-function}.
\end{proof}

Lemma~\ref{lem:spectral-control-deletion} formalizes the control-deletion step used
at two different points of the annular proof.  Its conclusion is membership
in the preceding filtered space; it does not assert that the shortened row
equals the original row.

\begin{lemma}[Spectral control deletion]\label{lem:spectral-control-deletion}
Let $r\ge1$, let $f\in\Sigma$, let $0\le s\le r-1$, and choose
$g_1,\ldots,g_s\in\Sigma\setminus\{a\}$.  For polynomials
$q_0,\ldots,q_{s+1}\in\C[X]$ and $\lambda\in\muRoots_m$, the row
\[
  u_fq_0(T)P_{g_1}q_1(T)\cdots
  P_{g_s}q_s(T)E_\lambda q_{s+1}(T)
\]
belongs to $\calM_{r-1}(f)$.
\end{lemma}

\begin{proof}
Expand the polynomials and the projector $E_\lambda$, which is a polynomial
in $T$ by \cref{lem:projectors}.  Every resulting row contains at most
$s\le r-1$ controls, so it lies in $\calL_{r-1}(f)$.  The displayed row is
also fixed by right multiplication by $E_\lambda$.  Indeed,
$E_\lambda q_{s+1}(T)=q_{s+1}(\lambda)E_\lambda$ and
$E_\lambda^2=E_\lambda$.  Hence the row lies in
$\im E_\lambda\subseteq R$.  Therefore it belongs
to $\calL_{r-1}(f)\cap R=\calM_{r-1}(f)$ by
\cref{lem:compatible-splitting}.
\end{proof}

With
\cref{lem:nilpotent-growth,lem:coefficient,lem:spectral-control-deletion}
in place, we prove the
proposition by
separating the projector sequences into two classes.  Repeated nonzero
eigenvalues give membership in $\calM_{r-1}$ term by term.  For sequences
without repetition, the vectors in
\eqref{eq:nilpotent-position-factorization} attached to the $N$-positions either
give another termwise reduction or satisfy the dimension increments in
\cref{lem:nilpotent-growth}.  For every remaining family, the sum of its
scalar coefficients is zero.

\begin{proof}[Proof of \cref{prop:annular-descent}]
\smallskip
\noindent\emph{Spectral expansion and zero terms.}
Insert \eqref{eq:projector-sum} at the $r+1$ spectral positions in every
summand of \eqref{eq:omega}.  Put
\[
  \mathcal P=\{\star\}\cup\muRoots_m.
\]
For $\alpha\in\mathcal P$, define $E_\alpha=E_N$ if $\alpha=\star$ and
$E_\alpha=E_\lambda$ if $\alpha=\lambda\in\muRoots_m$.  For
$\boldsymbol\alpha=(\alpha_0,\ldots,\alpha_r)\in\mathcal P^{r+1}$ and
$\mathbf h\in\mathcal A_r$, write
\begin{equation}\label{eq:spectral-path}
  \Pi_{\boldsymbol\alpha}(\mathbf h)=
  u_fE_{\alpha_0}T^{h_0}G_1E_{\alpha_1}T^{h_1}\cdots
  G_rE_{\alpha_r}T^{h_r+\ell}.
\end{equation}
Here $E_{\alpha_i}$ occupies spectral position $i$, whereas the control
$G_j$ joins positions $j-1$ and $j$.  Thus the segment between spectral
positions $p<q$ contains precisely the controls
$G_{p+1},\ldots,G_q$.  We call $i$ an \emph{$N$-position} if
$\alpha_i=\star$, and a \emph{nonzero-eigenvalue position} if
$\alpha_i\in\muRoots_m$.
Since all these projectors commute with $T$,
\begin{equation}\label{eq:path-expansion}
  \Omega_r=
  \sum_{\boldsymbol\alpha\in\mathcal P^{r+1}}
  \sum_{\mathbf h\in\mathcal A_r}
  \Pi_{\boldsymbol\alpha}(\mathbf h).
\end{equation}
All sums here are finite.  We call $\boldsymbol\alpha$ a
\emph{projector sequence} and one summand
$\Pi_{\boldsymbol\alpha}(\mathbf h)$ a \emph{spectral term}.  We work
modulo $\calM_{r-1}(f)$.  Every spectral projector preserves $\calZ$
because it is a polynomial in $T$, and every transition matrix preserves
$\calZ$.  Hence a spectral term whose projector sequence contains $1$ is
zero, since
$V_1=\Span_{\C}\set{[C]}$ has trivial intersection with $\calZ$.  A term with
$\alpha_r=\star$ is killed by the final $T^\ell$.  Likewise, a term with
$\alpha_i=\star$ and $h_i\ge\ell$ is zero.

For any spectral term not already proved to be zero, let
$i_1<\cdots<i_z<r$ be precisely the positions at which
$\alpha_i=\star$, and fix
$\boldsymbol\nu=(\nu_1,\ldots,\nu_z)$ with
$\nu_j=h_{i_j}\in\{0,\ldots,\ell-1\}$.  A \emph{spectral family}
is the set of spectral terms obtained by fixing
$\boldsymbol\alpha$ and $\boldsymbol\nu$ and varying the exponents
at the nonzero-eigenvalue positions subject to
$\mathbf h\in\mathcal A_r$.  These families form a finite partition of the
spectral terms not already proved to be zero.

\smallskip
\noindent\emph{Repeated nonzero eigenvalues.}
If a nonzero eigenvalue
$\lambda$ occurs at spectral positions $p<q$, let $\rho$ be the preceding
row factor, let $A_{p,q}$ be the intervening operator, and let $Z$ be the
following operator.  For a fixed exponent tuple, the corresponding term
then factors as
\[
  \rho E_\lambda A_{p,q}E_\lambda Z.
\]
Since $V_\lambda$ is one-dimensional, there is $\eta\in\C$ such that
\[
  E_\lambda A_{p,q}E_\lambda=\eta E_\lambda.
\]
If $\eta=0$, the spectral term is zero; if $\eta\ne0$, the term equals
$\eta\rho E_\lambda Z$.  The omitted intervening factor contains
exactly the $q-p\ge1$ controls between the two spectral positions.
The shortened expression still has a nonzero-eigenvalue projector as its
rightmost spectral projector.  It has at most $r-1$ controls, so
\cref{lem:spectral-control-deletion} puts it in $\calM_{r-1}$.  This is a
termwise reduction; no summation over exponents is being used.

\smallskip
\noindent\emph{Nilpotent positions and their dimension charge.}
Consider a spectral family to which the preceding termwise reductions do
not apply.  The eigenvalues occurring at its nonzero-eigenvalue positions
are pairwise distinct and different from $1$, and
$\alpha_r\in\muRoots_m$.  Put $s=\sum_j\nu_j$.  For a fixed
exponent tuple, factor the spectral term at its $j$th $N$-position as
\begin{equation}\label{eq:nilpotent-position-factorization}
  \Pi_{\boldsymbol\alpha}(\mathbf h)
  =x_j(\mathbf h)J^{\nu_j}B_j(\mathbf h).
\end{equation}
Here $x_j(\mathbf h)\in N$ is the product of all factors through the
projector $E_N$ at position $i_j$; it excludes the following factor
$T^{\nu_j}$.  The operator $B_j(\mathbf h)$ is the product of all
remaining factors after $T^{\nu_j}$.

For each $j$, let $\widetilde x_j$ be the same prefix with every free
exponent at a nonzero-eigenvalue position $q<i_j$ set equal to zero.  Since
$E_\lambda T^h=\lambda^hE_\lambda$, each such preceding exponent contributes
only the nonzero scalar $\lambda^h$.  The positions of all controls and
projectors, the exponents at the $N$-positions, and the resulting row vector
up to a nonzero scalar in $N$ are unchanged.  Consequently,
\begin{equation}\label{eq:nilpotent-direction-rescaling}
  x_j(\mathbf h)=c_j(\mathbf h)\widetilde x_j,
  \qquad
  c_j(\mathbf h)=
  \prod_{\substack{q<i_j\\ \alpha_q\in\muRoots_m}}
  \alpha_q^{h_q}\ne0.
\end{equation}
If there is no preceding nonzero-eigenvalue position, this product is empty
and equals $1$.  For $0\le j\le z$, define
\[
  H_j(\mathbf h)=
  \Span\set{x_i(\mathbf h)J^k:1\le i\le j,\ 0\le k<\ell},
  \qquad H_0(\mathbf h)=\{0\}.
\]
Because every scalar in \eqref{eq:nilpotent-direction-rescaling} is nonzero,
rescaling the generating rows does not change their span.  Hence
\begin{equation}\label{eq:nilpotent-span-rescaling}
  H_j(\mathbf h)=\widetilde H_j
  :=\Span\set{\widetilde x_iJ^k:1\le i\le j,\ 0\le k<\ell}.
\end{equation}
Thus this subspace is independent of the free exponents; we write it simply
as $H_j$.  We apply \cref{lem:nilpotent-growth} with $W=N$ and
$J=T|_N$.  Both alternatives in that lemma are therefore independent of
the free exponents at the nonzero-eigenvalue positions.
For the following termwise reduction, fix an exponent tuple and suppress
the argument $\mathbf h$ in $x_j(\mathbf h)$ and $B_j(\mathbf h)$.

Suppose the first alternative in \cref{lem:nilpotent-growth} occurs, and
write
\begin{equation}\label{eq:nilpotent-dependence}
  x_jJ^{\nu_j}
  =\sum_{i<j}\sum_{0\le k<\ell}\gamma_{i,k}x_iJ^k.
\end{equation}
The coefficients $\gamma_{i,k}$ may depend on the free exponents at the
nonzero-eigenvalue positions.  This causes no difficulty: the following
reduction is made separately for each exponent tuple, and only membership
of each resulting row in $\calM_{r-1}$ is used.
Multiplying \eqref{eq:nilpotent-dependence} by $B_j$ gives
\[
  x_jJ^{\nu_j}B_j
  =\sum_{i<j}\sum_{0\le k<\ell}\gamma_{i,k}x_iJ^kB_j.
\]
In the term indexed by $i$, all $i_j-i_i\ge1$ controls between the two
$N$-positions have been removed.  Moreover,
$x_iJ^k=x_iT^k$ because $x_i\in N$, and every remaining projector is a
polynomial in $T$.  The rightmost projector remains
$E_{\alpha_r}$ with $\alpha_r\in\muRoots_m$.  Each row on the right has at
most $r-1$ controls, so \cref{lem:spectral-control-deletion} places it in
$\calM_{r-1}$.  When $j=1$,
membership in $H_0=\{0\}$ simply says that the original term is zero.
Hence a spectral term not placed in $\calM_{r-1}$ by these reductions must
take the second alternative of \cref{lem:nilpotent-growth} at every
$N$-position.  Therefore
\begin{align}
  z+s
  &=\sum_{j=1}^z(1+\nu_j) \notag\\
  &\le\sum_{j=1}^z(\dim H_j-\dim H_{j-1})
    =\dim H_z\le\dim N=t. \label{eq:nilpotent-dimension-bound}
\end{align}
Also $z\le r$ and $1+\nu_j\le\ell$, so
\begin{equation}\label{eq:nilpotent-kappa-bound}
  z+s\le\ell r,
  \qquad
  z+s\le\kappa_r.
\end{equation}

\smallskip
\noindent\emph{Scalar annular cancellation.}
It remains to sum the scalar coefficient of a spectral family not covered
by the termwise reductions.  Fix one such family.  By
\cref{eq:nilpotent-direction-rescaling,eq:nilpotent-span-rescaling},
applicability of those reductions is independent of the still-free
exponents at the nonzero-eigenvalue positions.  Let
\[
  I_0=\set{i:\alpha_i\in\muRoots_m},
  \qquad \alpha_i=\lambda_i\quad(i\in I_0).
\]
Since $\alpha_r\in\muRoots_m$, we have $r\in I_0$.  Repeatedly using
$E_{\lambda_i}T^{h_i}=\lambda_i^{h_i}E_{\lambda_i}$ gives
\begin{equation}\label{eq:path-monomial}
  \Pi_{\boldsymbol\alpha}(\mathbf h)
  =\lambda_r^\ell
   \left(\prod_{i\in I_0}\lambda_i^{h_i}\right)
   v_{\boldsymbol\alpha,\boldsymbol\nu},
\end{equation}
where the row vector $v_{\boldsymbol\alpha,\boldsymbol\nu}$ is
independent of all remaining exponents.  If it is zero, the family
vanishes.  Otherwise
the eigenvalues $\lambda_i$ are pairwise distinct, so $Y_0$ below is a set
rather than a multiset.  Summing the monomial
in \eqref{eq:path-monomial} over nonnegative exponents of total exponent $k$
gives exactly $h_k(Y_0)$, where
\[
  Y_0=\set{\lambda_i:i\in I_0},
  \qquad Y=Y_0\cup\{1\}.
\]
There are $r+1-z$ nonzero-eigenvalue positions, hence
\begin{equation}\label{eq:Y-size}
  \abs Y=r+2-z.
\end{equation}
If $\abs Y>m$, such a remaining projector sequence cannot exist.  Otherwise,
in the notation of
\eqref{eq:roots-generating-function},
\begin{equation}\label{eq:complement-degree}
  c_Y:=\deg \Phi_Y=m-r+z-2\ge0.
\end{equation}
After removing the fixed nonzero factor $\lambda_r^\ell$, the annular
scalar coefficient of this spectral family is
\begin{align}
  \sum_{k=D_r^--s}^{D_r^+-s}h_k(Y_0)
  &=h_{D_r^+-s}(Y)-h_{D_r^--1-s}(Y). \label{eq:path-coefficient}
\end{align}
If $s>D_r^+$, the family is empty.  We may therefore set
$d=D_r^+-s\ge0$.  From
\cref{eq:annular-endpoints,eq:nilpotent-kappa-bound,eq:complement-degree},
\begin{equation}\label{eq:degree-gap}
  d-c_Y=\kappa_r+1-(z+s)\ge1.
\end{equation}

The two endpoint regimes now have the same coefficient calculation.  If
$\kappa_r\le r$, then $D_r^-=0$ and $d\le D_r^+\le m-1$; hence both
$h_{D_r^--1-s}(Y)$ and $h_{d-m}(Y)$ vanish.  If $\kappa_r>r$, then
$D_r^+=D_r^-+m-1$, so $D_r^--1-s=d-m$.  Thus in either case
\[
  h_d(Y)-h_{D_r^--1-s}(Y)
  =h_d(Y)-h_{d-m}(Y)
  =[x^d]\Phi_Y(x)=0.
\]
The coefficient identity comes from
\eqref{eq:roots-generating-function}, and its final value is zero because
\eqref{eq:degree-gap} gives $d>c_Y=\deg\Phi_Y$.  The negative-index
convention covers $d-m<0$, including the boundary $\kappa_r=r$.  At the
terminal filtration index $r=n-1$, the inequalities $\ell\ge1$ and
$n-1=m+t-1\ge t$ give $\kappa_r=t$, and hence
$D_r^-=D_r^+=0$.  Any remaining family would therefore have $s=d=0$,
contradicting the same degree gap $d>c_Y\ge0$.

We have exhausted all possibilities: a repeated nonzero eigenvalue or the
condition $x_jJ^{\nu_j}\in H_{j-1}$ gives membership in
$\calM_{r-1}$ term by term, whereas every remaining spectral family has
zero scalar coefficient.  Summing these finitely many families proves
\eqref{eq:annular-membership}.

\smallskip
\noindent\emph{The prefix bound.}
Finally, for a summand indexed by $\mathbf h\in\mathcal A_r$, prepend $f$
to its indexing block word and remove the terminal $a^\ell$.  The resulting
prefix has
length
\begin{equation}\label{eq:prefix-budget}
  1+r+\tau(\mathbf h)\le1+r+D_r^+=m+\kappa_r\le m+t=n.
\end{equation}
\end{proof}

\section{Positive induction and proof of short relative extension}\label{sec:short-extension-proof}

This section converts annular descent into the combinatorial extension
statement.  The proof is by contradiction.  Because every coefficient in
the annular sum is $+1$, nonpositive scalar pairings whose sum is zero must
vanish individually; this propagates orthogonality through the filtration.
Stabilization then removes the control bound from repeated leftmost-control
deletion, and a reset word contradicts the resulting invariance.

This is the positive-level counterpart of Dubuc's circular induction, which
combines the positive-coefficient relation of Lemma~4.4 with the
preimage-cardinality propagation of Lemma~4.5
~\cite[Lemmas~4.4--4.5]{Dubuc1998}.  Its base case uses the one-cluster
averaging identity familiar from B\'eal--Perrin and
Steinberg~\cite[Section~3]{BealPerrin2009}
\cite[Lemma~2]{Steinberg2011}; the additional input is the
annular descent relation modulo $\calM_{r-1}$.

For $S\subseteq C$ and $w\in\Sigma^*$, define
\begin{equation}\label{eq:mu}
  \mu_S(w)=\abs{Sw^{-1}\cap C}=[C]P_w[S]^{\mathsf T}.
\end{equation}
For the running example, $S=\{1,2\}$ and
\[
  x=(\mathtt{bbbaa})a^2=\mathtt{bbbaaaa}
\]
give $\mu_S(x)=3>2$; the displayed prefix $\mathtt{bbbaa}$ has length
$5=n$.  The proof below explains why such a positive jump always exists
with a prefix of length at most $n$.

\begin{proof}[Proof of \cref{thm:short-extension}]
Suppose, toward a contradiction, that a nonempty proper $S\subset C$
satisfies
\begin{equation}\label{eq:no-short-extension}
  \mu_S(wa^\ell)\le\abs S
  \qquad\text{for every }w\text{ with }\abs w\le n.
\end{equation}

We prove simultaneously for every $f\in\Sigma$ that
\begin{equation}\label{eq:orthogonality}
  v[S]^{\mathsf T}=0
  \qquad(v\in\calM_r(f))
\end{equation}
first for $0\le r\le n-2$.

\smallskip
\noindent\emph{Base of the filtration.}
For the base case, use the one-cluster averaging identity
\begin{equation}\label{eq:average-matrix}
  B:=\sum_{j=0}^{m-1}T^{\ell+j}=\one[C].
\end{equation}
Indeed, each row on the left visits every state of $C$ exactly once.  This is
the unnormalized matrix form of the usual uniform average over the
$a$-cycle; omitting the factor $1/m$ preserves the positive coefficient
$+1$ of every summand.
Since $u_f\one=0$ and $B=\one[C]$, we have $u_fB=0$.  On the other hand,
for $0\le j<m$,
\begin{align}
  u_fT^{\ell+j}[S]^{\mathsf T}
  &=\mu_S(fa^{\ell+j})-\mu_S(a^{\ell+j}) \notag\\
  &\le0. \label{eq:base-nonpositive}
\end{align}
Here $fa^{\ell+j}=(fa^j)a^\ell$, and the prefix $fa^j$ preceding the
displayed terminal $a^\ell$ has length $j+1\le m\le n$.  Thus the first
term is at most $\abs S$ by
\eqref{eq:no-short-extension}; the second is exactly $\abs S$ because
$[C]T^k=[C]$ for every $k\ge0$.  The $m$ nonpositive quantities in
\eqref{eq:base-nonpositive} sum to zero.  Therefore each is zero.  Finally,
\[
  \calM_0(f)
  =\Span\set{u_fT^{\ell+j}:j\ge0}
  =\Span\set{u_fT^{\ell+j}:0\le j<m},
\]
where the second equality follows from $T^m|_R=I_R$.  This proves
\eqref{eq:orthogonality} for $r=0$.  If
$f=a$, then $u_f=0$ and every assertion is vacuous.

\smallskip
\noindent\emph{Annular induction.}
We now carry out the induction step.  Fix $1\le r\le n-2$ and assume
\eqref{eq:orthogonality} is known at filtration index $r-1$ for every
letter $f\in\Sigma$.  We first record the \emph{leftmost-control deletion
identity} used below.
Let $x$ be a word ending in $a^\ell$ with at most $r$ letters different
from $a$.  If $x$ has a control letter, write
\[
  x=a^igz,
  \qquad g\ne a,
\]
where $g$ is the leftmost one.  Since $[C]T^i=[C]$,
\begin{align}
  \mu_S(x)-\mu_S(z)
  &=[C](P_g-I)P_z[S]^{\mathsf T} \\
  &=u_gP_z[S]^{\mathsf T}=0. \label{eq:leftmost-control-deletion}
\end{align}
The last equality holds because $z$ ends in $a^\ell$, has at most $r-1$
controls, and therefore $u_gP_z\in\calM_{r-1}(g)$.  Repeating the deletion
until a pure power of $a$ remains gives
\begin{equation}\label{eq:control-deletion-conclusion}
  \mu_S(x)=\abs S
\end{equation}
for every such $x$.

Now fix $f\in\Sigma$ and
$g_1,\ldots,g_r\in\Sigma\setminus\{a\}$.  Every word $x$ indexing a term
of the annular sum \eqref{eq:omega} ends in $a^\ell$ and has $r$ controls.  By
\eqref{eq:prefix-budget}, $fx=wa^\ell$ with $\abs w\le n$.  Hence
\begin{equation}\label{eq:annular-nonpositive}
  u_fP_x[S]^{\mathsf T}
   =\mu_S(fx)-\mu_S(x)\le0
\end{equation}
by \cref{eq:no-short-extension,eq:control-deletion-conclusion}.  Yet
$\Omega_r\in\calM_{r-1}(f)$ by
\cref{prop:annular-descent}.  The induction hypothesis therefore gives
$\Omega_r[S]^{\mathsf T}=0$.  By the definition of $\Omega_r$, this is the
sum of the quantities in \eqref{eq:annular-nonpositive}.  Every coefficient
in the annular sum is $+1$; therefore every individual quantity is zero.

It remains to verify orthogonality for the chosen generators of $\calM_r$
whose represented rows do not already belong to $\calM_{r-1}$.  Choose the
actual-word generating family
$((w_\gamma,x_\gamma))_{\gamma\in\Gamma}$ from \cref{lem:generation}, whose
representing words end in $a^\ell$.  If $w_\gamma$ has fewer than $r$
controls, then its represented row $x_\gamma$ already belongs to
$\calM_{r-1}$.  If $w_\gamma$ has exactly $r$ controls, let
$\mathbf h_\gamma=(h_0,\ldots,h_r)$ be its unique exponent tuple in the
block form \eqref{eq:block-form}.  Then
the proof that the total-exponent interval is nonempty also gives
$D_r^+\ge1$ for $1\le r\le n-2$, and
\eqref{eq:generator-total-exponent-bound} says
\begin{equation}\label{eq:generator-below-upper-endpoint}
  \tau(\mathbf h_\gamma)\le D_r^+-1.
\end{equation}
If $\tau(\mathbf h_\gamma)\ge D_r^-$, then $w_\gamma$ indexes an annular
term, so $x_\gamma[S]^{\mathsf T}=0$.  If
$\tau(\mathbf h_\gamma)<D_r^-$, then $D_r^->0$ and necessarily
$\kappa_r>r$.  The interval
$[D_r^--\tau(\mathbf h_\gamma),D_r^+-\tau(\mathbf h_\gamma)]$ consists of
$m$ consecutive positive integers, so it contains a unique
$\varkappa\in m\mathbb Z_{>0}$.  The row $x_\gamma$ belongs to
$\calM_r\subseteq R$, and $T^m|_R=I_R$, so the word
$w_\gamma a^\varkappa$ represents
$x_\gamma T^\varkappa=x_\gamma$.  Its total exponent is
$\tau(\mathbf h_\gamma)+\varkappa\in[D_r^-,D_r^+]$, and its controls are
unchanged; hence it indexes an annular term.  Therefore
$x_\gamma[S]^{\mathsf T}=0$ in this case as well.
The rows $x_\gamma$ span $\calM_r$, so linearity proves
\eqref{eq:orthogonality} at index $r$.  The induction proves it for
$0\le r\le n-2$.

\smallskip
\noindent\emph{Stabilization and unrestricted control deletion.}
We next pass from these finitely many filtration indices to arbitrary words.
By \eqref{eq:gen-total}, after multiplying by $T^\ell$, every $\calM_r(f)$
has an actual-word generating family whose representing words end in
$a^\ell$ and have length at most
\[
  d_r+k_r-1+\ell
  \le(m-1)+t-1+\ell=n+\ell-2.
\]
Because the last $\ell$ letters are $a$, each representing word has at most
$n-2$ controls; hence its represented row belongs to $\calM_{n-2}(f)$.
Thus
$\calM_r(f)\subseteq\calM_{n-2}(f)$ for every $r$.  If $r\ge n-2$, the
reverse inclusion follows from \eqref{eq:filtration-inclusions}.  Hence
\begin{equation}\label{eq:stabilization}
  \calM_r(f)=\calM_{n-2}(f)
  \qquad(r\ge n-2).
\end{equation}
The leftmost-control deletion argument now applies without a bound on the
number of controls.  If $x$ has no controls, then it is a power of $a$ and
$\mu_S(x)=\abs S$ because $a$ permutes $C$.  Otherwise, in the notation
$x=a^igz$ used in
\eqref{eq:leftmost-control-deletion}, if $x$ has $c$ controls, then $z$ has
$c-1$ controls and hence $u_gP_z\in\calM_{c-1}(g)$.  If $c-1\le n-2$, its
pairing with $[S]^{\mathsf T}$ vanishes by \eqref{eq:orthogonality}; if
$c-1>n-2$, the same conclusion follows after applying
\eqref{eq:stabilization}.  Thus
\begin{equation}\label{eq:all-terminal-preserve}
  \mu_S(x)=\abs S
\end{equation}
for every word $x$ ending in $a^\ell$, with no restriction on its length.

\smallskip
\noindent\emph{Contradiction from synchronization.}
The synchronizing hypothesis supplies a reset word $z$; let $q$ be its
reset state.  The state $q\mathbin{\cdot}a^\ell$ lies in $C$, and $a$ acts
on $C$ as an $m$-cycle.  Since $S$ is nonempty, there is therefore an
integer $0\le j<m$ such that
$q\mathbin{\cdot}a^{\ell+j}\in S$.  The word
$za^{\ell+j}=(za^j)a^\ell$ ends in $a^\ell$, but
\[
  \mu_S(za^{\ell+j})=m>\abs S,
\]
contradicting \eqref{eq:all-terminal-preserve}.  This proves
\cref{thm:short-extension}.
\end{proof}

\paragraph{The running example: a short relative extension at the prefix
bound.}
In \cref{ex:running}, take $S=\{1,2\}$ and
$w=\mathtt{bbbaa}$.  For
\[
  x=wa^2=\mathtt{bbbaaaa},
\]
the three cycle states $0,1,2$ are sent to $1,2,2$, respectively.  Hence
\[
  Sx^{-1}\cap C=C,
\]
so this step increases the relative-preimage cardinality from $2$ to $3$.
The chosen prefix has length $\abs w=5=n$, which is the allowed length in
\cref{thm:short-extension}.  Notice
that the same string $\mathtt{bbbaa}$ reappears in \cref{sec:singleton} as the
whole relative extending word $w_0$ from a singleton, where it already
contains the prescribed terminal $a^2$.  Here the same string instead plays
the role of the bounded prefix $w$, and a further terminal $a^2$ occurs in
$x=wa^2$.

\section{The singleton start}\label{sec:singleton}

To initiate the relative-extension chain efficiently, we need one extension
from a singleton.  The centered-column argument of Steinberg
~\cite[Claim~2]{Steinberg2011PrimeCycle} uses neither primality nor cyclotomic
irreducibility.  In this section transition matrices act on \emph{column}
vectors by left multiplication.  For nonempty $S\subseteq C$, set
\begin{equation}\label{eq:gamma}
  \gamma_S=[S]^{\mathsf T}-\frac{\abs S}{m}[Q]^{\mathsf T}.
\end{equation}
Equation \eqref{eq:counting} gives
\begin{equation}\label{eq:gamma-count}
  [C]P_w\gamma_S=\abs{Sw^{-1}\cap C}-\abs S.
\end{equation}

\Needspace{7\baselineskip}
Thus strict growth from a singleton is equivalent to finding a positive
pairing with one of these centered columns.  Proposition~\ref{prop:singleton-start} obtains
such a pairing from the strict growth of a sequence of invariant subspaces.
Its bound $t+1+\ell$ is precisely the initial cost needed in the final reset
accounting.

\begin{proposition}[Singleton start]\label{prop:singleton-start}
Let $A=\langle Q,\Sigma,\delta\rangle$ be a synchronizing one-cluster
automaton with respect to $a\in\Sigma$.  Let $C$ be its $a$-cycle, let
$\ell$ be its level with respect to $a$, assume $m=\abs C\ge2$ and
$\ell\ge1$, and put $t=\abs Q-m$.  There exist
$q\in C$ and a word $w_0\in\Sigma^*$ such that
\begin{equation}\label{eq:singleton-start}
  \abs{w_0}\le t+1+\ell,
  \qquad
  \abs{\{q\}w_0^{-1}\cap C}>1.
\end{equation}
\end{proposition}

Unlike \cref{thm:short-extension}, this proposition asserts only the
existence of a suitable singleton $\{q\}$ rather than a statement for every
singleton.  In return it gives the smaller total bound $t+1+\ell$, which is
the economical initialization needed by the final relative-extension
construction.

\begin{proof}
For $q\in C$, put
\[
  z_q=T^\ell\gamma_{\{q\}},
  \qquad
  W=\Span\set{z_q:q\in C},
  \qquad
  H_C=\set{x\in\C^{n\times1}:[C]x=0}.
\]
$H_C=\ker[C]$ is a hyperplane of the column space.  It is distinct from the
row space $\calZ$ in \eqref{eq:row-zero-sum}, and it is not asserted to be
invariant under the letter matrices.
Write $\delta_p$ for the $p$th standard basis column of $\C^C$ and
$\one_C$ for the all-one column in $\C^C$.  For each $q\in C$, let $p_q\in C$ be the
unique state satisfying $p_q\mathbin{\cdot}a^\ell=q$.  If
$z_q|_C\in\C^C$ denotes the
restriction of $z_q$ to its coordinates indexed by $C$, then
\[
  z_q|_C=\delta_{p_q}-\frac1m\one_C.
\]
As $q$ ranges over $C$, so does $p_q$.  Hence these columns span the
$(m-1)$-dimensional zero-sum subspace of $\C^C$, and $\dim W\ge m-1$.  On
the other hand,
\[
  \sum_{q\in C}z_q
  =T^\ell\bigl([C]^{\mathsf T}-[Q]^{\mathsf T}\bigr)=0,
\]
because both $C(a^\ell)^{-1}$ and $Q(a^\ell)^{-1}$ equal $Q$.
Equivalently, \eqref{eq:counting} gives
$T^\ell[C]^{\mathsf T}=T^\ell[Q]^{\mathsf T}=[Q]^{\mathsf T}$.
Thus $\dim W=m-1$.  Also $[C]T^\ell=[C]$ and
$[C]\gamma_{\{q\}}=0$, so $W\subseteq H_C$.

Define
\[
  \Sigma^*W=
  \Span\set{P_uz:u\in\Sigma^*,\ z\in W}.
\]
It is the smallest subspace containing $W$ and invariant under left
multiplication by every letter matrix $P_g$, $g\in\Sigma$.  This subspace is
not contained in $H_C$.  Here the synchronizing hypothesis enters.  Fix
$q\in C$ and let $p\in C$ be the unique
state with $p\mathbin{\cdot}a^\ell=q$.  Starting with any
reset word, append $a^{\ell+j}$ for a suitable $0\le j<m$; the resulting
word $y$ resets all states to $p$.  Then
\[
  [C]P_yz_q=m\left(1-\frac1m\right)=m-1\ne0.
\]
Now form the ascending chain
\[
  W_j=\Span\set{P_uz_q:q\in C,\ \abs u\le j}.
\]
If $W_j=W_{j+1}$, then $W_j$ is invariant under every letter and contains
$W$, so it contains $\Sigma^*W$.  Therefore, as long as
$W_j\subseteq H_C$, the chain grows strictly.  Moreover,
$\Sigma^*W=\bigcup_{j\ge0}W_j$, because every word has finite
length.  Thus the least index
\[
  j_0=\min\set{j\ge0:W_j\nsubseteq H_C}
\]
exists.  Since $\dim W_0=m-1$, $\dim H_C=n-1$, and every inclusion
$W_j\subsetneq W_{j+1}$ is strict for $j<j_0$, we have
$j_0\le(n-1)-(m-1)+1=t+1$.  The displayed spanning family for $W_{j_0}$
and the fact that $H_C$ is a subspace now give a word $u$ and a state $q\in C$
such that
\[
  \abs u\le j_0\le t+1,
  \qquad
  P_uz_q\notin H_C.
\]
Equivalently, $\beta_q:=[C]P_uz_q\ne0$.  For this fixed $u$, all
$\beta_q$ are real and
\[
  \sum_{q\in C}\beta_q=[C]P_u\sum_{q\in C}z_q=0.
\]
Thus at least one of them is positive.  For that $q$, \eqref{eq:gamma-count}
gives
\[
  0<\beta_q
   =\abs{\{q\}(a^\ell)^{-1}u^{-1}\cap C}-1.
\]
Taking $w_0=ua^\ell$ gives
$\abs{w_0}=\abs u+\ell\le t+1+\ell$ and proves
\eqref{eq:singleton-start}.
\end{proof}

\paragraph{The running example: the centered singleton step.}
Take $u=\mathtt{bbb}$ and $q=0$.  Then
$w_0=ua^2=\mathtt{bbbaa}$, and the three cycle states $0,1,2$ are sent by
$w_0$ to $2,0,0$.  Thus
\[
  \{0\}w_0^{-1}\cap C=\{1,2\}.
\]
For the singleton targets $q=0,1,2$, respectively, the centered values
$\beta_q$ are $1,-1,0$; their zero sum and positive entry are visible
without linear algebra.  Moreover,
$\abs{w_0}=5=t+1+\ell$, so this chosen word attains the length allowed by
\cref{prop:singleton-start}.

\section{Reset accounting and proof of the main theorem}\label{sec:reset-proof}

We now combine the singleton start with at most $m-2$ short relative
extensions,
keeping track of the reversed preimage order and of the total word
length.

\begin{proof}[Proof of \cref{thm:reset-bound}]
If $m=1$, then $C$ is a singleton and
$Q\mathbin{\cdot}a^\ell=C$, so $a^\ell$ is a reset word.  Thus
$\rt(A)\le\ell=(m-1)(n-1)+m\ell$.  The inequality
$\ell\le(n-1)^2$ follows from \eqref{eq:level-le-t}; when $n=1$, both sides
are zero, and when $n\ge2$, we have $\ell\le n-1\le(n-1)^2$.

Suppose $m\ge2$.  If $\ell=0$, then $C=Q$, $m=n$, and Dubuc's theorem gives
$\rt(A)\le(n-1)^2=(m-1)(n-1)+m\ell$
~\cite[Proposition~4.6]{Dubuc1998}.  It remains to consider
$m\ge2$ and $\ell\ge1$.

Choose $q$ and $w_0$ as in
\cref{prop:singleton-start}, and let
\[
  S_0=\{q\}w_0^{-1}\cap C.
\]
The update $S\mapsto Su^{-1}\cap C$ and the left-prepending of each new
factor are the same as in Volkov's algorithm \textsc{RelativeExtension}
~\cite[Section~3.4]{Volkov2022Survey}.  The quantitative inputs differ:
Volkov's generic pseudocode starts from an arbitrary singleton and finally
uses $a^t$.  Here the initial word is the shorter singleton word $w_0$;
every later factor has the Kisielewicz--Kowalski--Szyku\l a form
$v_i a^\ell$; and the minimal power $a^\ell$ maps $Q$ onto $C$.  The factors
accumulate on the left because
full preimages satisfy
$(Su^{-1})v^{-1}=S(vu)^{-1}$; the reversed product order is an algebraic
feature of preimages, not a different notion of synchronization.
Thus $\abs{S_0}\ge2$.  Whenever $S_i\ne C$, apply
\cref{thm:short-extension} to choose
\[
  x_{i+1}=v_{i+1}a^\ell,
  \qquad \abs{x_{i+1}}\le n+\ell,
\]
and put
\[
  S_{i+1}=S_ix_{i+1}^{-1}\cap C.
\]
Continue until $S_k=C$.  The cardinality increases strictly, so
$0\le k\le m-2$; the case $k=0$ means that $S_0=C$ already.  Put
$X_0=\varepsilon$ and $X_i=x_i x_{i-1}\cdots x_1$ for $i\ge1$.  We claim
that
\[
  S_i\subseteq
  \{q\}(X_iw_0)^{-1}
  \qquad(0\le i\le k).
\]
For $i=0$, we have $X_0=\varepsilon$, so the claimed inclusion is precisely
$S_0\subseteq\{q\}w_0^{-1}$, which follows from the definition of $S_0$.
If it
holds at $i$, then
\begin{align*}
  S_{i+1}
  &\subseteq S_ix_{i+1}^{-1}\\
  &\subseteq
    \bigl(\{q\}(X_iw_0)^{-1}\bigr)x_{i+1}^{-1}\\
  &=\{q\}(x_{i+1}X_iw_0)^{-1}
   =\{q\}(X_{i+1}w_0)^{-1},
\end{align*}
which proves the claim by induction.  Since $S_k=C$, this gives
\[
  C\subseteq
  \{q\}(X_k w_0)^{-1}.
\]
Equivalently,
\[
  C\mathbin{\cdot}(X_k w_0)=\{q\}.
\]
Since $Q\mathbin{\cdot}a^\ell=C$, the word
\begin{equation}\label{eq:final-reset-word}
  a^\ell X_k w_0
\end{equation}
is a reset word.  Its length is at most
\begin{align}
  \abs{a^\ell X_k w_0}
  &\le \ell+k(n+\ell)+(t+1+\ell) \notag\\
  &\le \ell+(m-2)(n+\ell)+(t+1+\ell) \notag\\
  &=(m-1)(n-1)+m\ell. \label{eq:accounting}
\end{align}
Finally, by \eqref{eq:level-le-t},
\begin{align}
  (n-1)^2-\bigl((m-1)(n-1)+m\ell\bigr)
  &=m(t-\ell)+t(t-1)\ge0. \label{eq:cerny-slack}
\end{align}
Here $t-\ell\ge0$ by \eqref{eq:level-le-t}, and $t\ge1$ because
$\ell\ge1$.  This proves \cref{thm:reset-bound} in the remaining case and
completes the proof.
\end{proof}

\paragraph{The running example: the complete relative-extension chain.}
The singleton step in \cref{prop:singleton-start} gives
$q=0$, $w_0=\mathtt{bbbaa}$, and
$S_0=\{1,2\}$.  The short relative extension from
\cref{sec:short-extension-proof} is
\[
  x_1=\mathtt{bbbaaaa},
  \qquad S_0x_1^{-1}\cap C=C.
\]
Equivalently,
\[
  Q\xrightarrow{\ a^2\ }C
   \xrightarrow{\ x_1\ }\{1,2\}
   \xrightarrow{\ w_0\ }\{0\}.
\]
Thus the word supplied by the proof is
\[
  a^2x_1w_0=\mathtt{aabbbaaaabbbaa}.
\]
It has length $2+7+5=14$, equals the right-hand side of the refined estimate
$(m-1)(n-1)+m\ell=14$, and sends every state to $0$.  The displayed equality
concerns the length delivered by the construction; no minimality assertion is
made here.  The general \v{C}ern\'y bound for five states is $16$.

\section{Remarks on the annular mechanism}\label{sec:remarks}

The proof is entirely finite-dimensional, but three features of the
annular relation are worth emphasizing.

First, the two generation bounds in \cref{lem:generation} play different
roles.  The total-dimension bound contains
$k_r=\dim(\calL_r\cap N)$.  In the second construction, every new control
uses the annihilating polynomial $X^\ell\psi_r(X)$ and therefore adds at
most $\ell$ beyond the growth of the $R$-component.  Taking the minimum of
the two bounds produces $\kappa_r=\min(t,\ell r)$.  For a generator with
$r$ controls and exponent tuple $\mathbf h_\gamma$, the strict inequality
$\tau(\mathbf h_\gamma)\le D_r^+-1$ permits the case
$\tau(\mathbf h_\gamma)<D_r^-$ to be shifted into the total-exponent
interval by a positive multiple of $m$.

Second, the relevant nilpotent dimension bound is $z+s$, not merely the
number $z$ of $N$-positions.  For every $j$ at which the term is not reduced
to $\calM_{r-1}$, \cref{lem:nilpotent-growth} says that the quotient classes
\[
  x_j+H_{j-1},\ x_jJ+H_{j-1},\ldots,
  x_jJ^{\nu_j}+H_{j-1}
\]
are linearly independent.  Summing the corresponding dimension increments
gives \eqref{eq:nilpotent-dimension-bound}, which is the estimate used in
\eqref{eq:degree-gap}.

Third, when $\kappa_r>r$, the total-exponent interval contains exactly $m$
integers:
\[
  D_r^+-D_r^-+1=m.
\]
The roots-of-unity generating function then compares coefficient indices
$d$ and $d-m$.  Replacing the lower endpoint $D_r^-$ by zero would remove
the $h_{d-m}(Y)$ term in \eqref{eq:path-coefficient} and leave an uncancelled
boundary coefficient.  Thus both endpoints in
\eqref{eq:annular-exponent-set} are required by the interaction between the
powers of $J$ on $N$ and the roots-of-unity spectrum on $R$.

For $\ell=1$, we have $J=0$, so there are no nontrivial nilpotent iterates.
For $\ell=2$, the first shifted total-exponent interval, at filtration index
$r=1$, is
\[
  1\le h_0+h_1\le m.
\]
No step of the proof depends on the cycle length being prime.

\section{Sharpness constructions for the refined bound}
\label{sec:sharpness}

This section gives, for every order $n\ge4$, a strongly connected binary
automaton for which the first inequality in \eqref{eq:reset-bound} is an
equality.  The definition, the one-cluster parameters, and an upper bound are
uniform in $n$.  The lower bounds are proved separately.  In even order, the
second letter is not injective and two letter-counting invariants apply.  In
odd order, that letter is a permutation, and the lower bound instead comes
from a forced detour in a graph of cyclic intervals.

The examples below have $m=2$ and $\ell=n-2$, so their reset thresholds grow
linearly with $n$.  They establish sharpness of the parameter-dependent bound
$(m-1)(n-1)+m\ell$ along this family.  They do not assert equality in the
coarser estimate $(m-1)(n-1)+m\ell\le(n-1)^2$.

\begin{definition}[The family $\mathcal F_n$]\label{def:sharp-family}
For an integer $n\ge4$, let $\mathcal F_n$ be the automaton with state set
$Q_n=\{0,1,\ldots,n-1\}$ and alphabet $\{a,b\}$.  Its transitions are
\[
  0\mathbin{\cdot}a=1,
  \qquad
  i\mathbin{\cdot}a=i-1
  \quad(1\le i\le n-1),
\]
and
\[
  i\mathbin{\cdot}b=i+1
  \quad(0\le i<n-1),
  \qquad
  (n-1)\mathbin{\cdot}b=
  \begin{cases}
    1,&\text{if $n$ is even},\\
    0,&\text{if $n$ is odd}.
  \end{cases}
\]
\end{definition}

The two letter actions are displayed schematically in
\cref{fig:sharp-family-a,fig:sharp-family-b}.  The arrows through each
ellipsis continue through the consecutively labelled states; when $n=4$,
there is no intermediate state between $2$ and $n-1=3$.

\begin{figure}[htbp]
\centering
\begin{tikzpicture}[x=1.55cm,y=1cm]
  \node[oc cycle state] (fa0) at (0,0)   {$0$};
  \node[oc cycle state] (fa1) at (1,0)   {$1$};
  \node[oc state]       (fa2) at (2,0)   {$2$};
  \node                 (fad) at (3.05,0) {$\cdots$};
  \node[oc state]       (fan) at (4.25,0) {$n-1$};

  \draw[oc transition] (fa1) -- (fa0);
  \draw[oc transition] (fa0) to[bend right=45] (fa1);
  \draw[oc transition] (fa2) -- (fa1);
  \draw[oc transition] (fan) -- (fad);
  \draw[oc transition] (fad) -- (fa2);

  \node[font=\small] at (0.5,0.55) {$C=\{0,1\}$};
  \node[font=\small] at (3.05,0.55) {$\ell=n-2$};
\end{tikzpicture}
\caption{The common functional digraph of $a$ in $\mathcal F_n$.  The
shaded vertices form the $a$-cycle $0\to1\to0$, and the remaining arrows
form the tail $n-1\to\cdots\to2\to1$.}
\label{fig:sharp-family-a}
\end{figure}

\begin{figure}[htbp]
\centering
\begin{minipage}[t]{0.48\linewidth}
\centering
\vspace{0pt}
\begin{tikzpicture}[x=1.12cm,y=1cm]
  \node[oc state] (fbe0) at (0,0)    {$0$};
  \node[oc state] (fbe1) at (1,0)    {$1$};
  \node[oc state] (fbe2) at (2,0)    {$2$};
  \node           (fbed) at (3.05,0) {$\cdots$};
  \node[oc state] (fben) at (4.25,0) {$n-1$};

  \draw[oc transition] (fbe0) -- (fbe1);
  \draw[oc transition] (fbe1) -- (fbe2);
  \draw[oc transition] (fbe2) -- (fbed);
  \draw[oc transition] (fbed) -- (fben);
  \draw[oc transition] (fben) to[bend right=42] (fbe1);
\end{tikzpicture}

\smallskip
\small\textbf{Even $n$.}
\end{minipage}
\hfill
\begin{minipage}[t]{0.48\linewidth}
\centering
\vspace{0pt}
\begin{tikzpicture}[x=1.12cm,y=1cm]
  \node[oc state] (fbo0) at (0,0)    {$0$};
  \node[oc state] (fbo1) at (1,0)    {$1$};
  \node[oc state] (fbo2) at (2,0)    {$2$};
  \node           (fbod) at (3.05,0) {$\cdots$};
  \node[oc state] (fbon) at (4.25,0) {$n-1$};

  \draw[oc transition] (fbo0) -- (fbo1);
  \draw[oc transition] (fbo1) -- (fbo2);
  \draw[oc transition] (fbo2) -- (fbod);
  \draw[oc transition] (fbod) -- (fbon);
  \draw[oc transition] (fbon) to[bend right=42] (fbo0);
\end{tikzpicture}

\smallskip
\small\textbf{Odd $n$.}
\end{minipage}
\caption{The parity-dependent functional digraphs of $b$ in $\mathcal F_n$.
For even $n$, state $0$ feeds into the cycle
$1\to2\to\cdots\to n-1\to1$.  For odd $n$, all states form the cycle
$0\to1\to\cdots\to n-1\to0$.  The return arrow is the only
parity-dependent transition.}
\label{fig:sharp-family-b}
\end{figure}

The $a$-skeleton consists of the cycle $0\to1\to0$ and the tail
\[
  n-1\longrightarrow n-2\longrightarrow\cdots
  \longrightarrow2\longrightarrow1.
\]
Thus $C=\{0,1\}$, $m=2$, and $\ell=n-2$.  If $n$ is odd, then $b$ is a
cycle on all of $Q_n$.  If $n$ is even, then $b$ restricts to a cycle on
$\{1,\ldots,n-1\}$.  In the even case every state reaches $1$ by a power of
$a$.  State $1$ reaches every state in $\{1,\ldots,n-1\}$ by a power of $b$
and reaches $0$ by the transition $1\mathbin{\cdot}a=0$.  Hence
$\mathcal F_n$ is strongly connected in both parity cases.

The next lemma supplies the common upper bound.  Its proof also explains why
the chosen exceptional $b$-transition depends on the parity of $n$.

\begin{lemma}[A common reset word]\label{lem:sharp-common-word}
For every integer $n\ge4$, the word
\[
  w_n=a^{n-2}b^{n-1}a^{n-2}
\]
resets $\mathcal F_n$ to state $1$.  In particular,
$\rt(\mathcal F_n)\le3n-5$.
\end{lemma}

\begin{proof}
The first block maps the state set onto the $a$-cycle:
\[
  Q_n\mathbin{\cdot}a^{n-2}=\{0,1\}.
\]
For the middle block, the definition of $b$ gives
\[
  \{0,1\}\mathbin{\cdot}b^{n-1}
  =
  \begin{cases}
    \{n-1,1\},&\text{if $n$ is even},\\
    \{n-1,0\},&\text{if $n$ is odd}.
  \end{cases}
\]
In both cases, $(n-1)\mathbin{\cdot}a^{n-2}=1$.  If $n$ is even, then
$1\mathbin{\cdot}a^{n-2}=1$ because $n-2$ is even.  If $n$ is odd, then
$0\mathbin{\cdot}a^{n-2}=1$ because $n-2$ is odd.  Therefore the last block
maps either displayed pair to $\{1\}$.  The length of $w_n$ is
$2(n-2)+(n-1)=3n-5$.
\end{proof}

For a word $u\in\{a,b\}^*$, let $\abs u_a$ and $\abs u_b$ denote the
numbers of occurrences of $a$ and $b$ in $u$, respectively.  The next lemma
records two facts common to both parity cases.  The terminal-state conclusion
will initialize both reverse-preimage arguments.  The letter-count conclusion
will be combined with the even-order potential below; it also applies in odd
order, although the stronger odd-order argument counts the entire word at
once.

\begin{lemma}[Common terminal state and parity obstruction]
\label{lem:sharp-common-obstruction}
For every integer $n\ge4$, every shortest reset word of $\mathcal F_n$ resets
the automaton to state $1$.  Moreover, every reset word $u$ satisfies
\begin{equation}\label{eq:common-b-count}
  \abs u_b\ge n-1.
\end{equation}
\end{lemma}

\begin{proof}
Let $u=vg$ be a shortest reset word, where $g\in\{a,b\}$ is its last letter.
The word $v$ is not a reset word, so $Q_n\mathbin{\cdot}v$ contains at least
two distinct states with the same image under $g$.  The only two distinct
states merged by $a$ are $0$ and $2$, and their common image is $1$.  When
$n$ is even, the only two distinct states merged by $b$ are $0$ and $n-1$,
again with common image $1$.  When $n$ is odd, $b$ is a permutation and
merges no two distinct states.  These three cases show that $u$ resets to $1$.

For the letter count, track the images of the pair $\{0,1\}$ until its first
merger, and record whether the two numerical labels have equal or opposite
parity.  They initially have opposite parity.  The letter $a$ reverses the
parity of every state.  The letter $b$ also reverses parity
except at its final transition
\[
  (n-1)\mathbin{\cdot}b=
  \begin{cases}
    1,&\text{if $n$ is even},\\
    0,&\text{if $n$ is odd},
  \end{cases}
\]
whose two endpoints have the same parity.  Thus the relation ``equal parity''
versus ``opposite parity'' changes only when a $b$-step has exactly one
tracked state at $n-1$.

At least one such exceptional $b$-step occurs no later than the first merger.
Choose the first one.  If the merger uses $a$,
then immediately beforehand the tracked states are $0$ and $2$ and have the
same parity, so their parity relation changed earlier.
If the merger uses $b$, then $n$ is even and the tracked states form
$\{0,n-1\}$; the merging step itself is the required exceptional step.

Immediately before the exceptional step, one tracked state is at $n-1$.
Before the first merger, the two tracked states are distinct and their maximum
numerical label is at least $1$.  On such a pair, applying $a$ does not
increase the maximum: the image of $0$ is $1$, and every positive label
decreases by one.  Applying $b$ increases the maximum by at most one.  Since
the initial maximum is $1$, at least $n-2$ occurrences of $b$ are needed
before one tracked state can reach $n-1$.  The exceptional step is one
further occurrence, which proves \eqref{eq:common-b-count}.
\end{proof}

\subsection{Even orders: two letter-counting invariants}
\label{subsec:sharp-even}

Fix an even integer $n\ge4$.  The common obstruction in
\cref{lem:sharp-common-obstruction} forces at least $n-1$ occurrences of $b$.
The remaining invariant uses full preimages to force at least $2n-4$
occurrences of $a$ in a shortest reset word.

For the second count, use the cyclic order
\begin{equation}\label{eq:even-cyclic-order}
  0,2,4,\ldots,n-2,1,3,5,\ldots,n-1.
\end{equation}
A \emph{cyclic interval} in this order is a set of consecutive entries, with
indices read cyclically; the empty set and $Q_n$ are also called cyclic
intervals.  For $S\subseteq Q_n$, write $s_i=1$ if $i\in S$ and $s_i=0$
otherwise.  Define
\[
  r(S)=\sum_{i=1}^{n-1}s_i,
  \qquad
  \iota(S)=
  \begin{cases}
    1,&\text{if $s_0=1$ and $s_{n-1}=0$},\\
    0,&\text{otherwise},
  \end{cases}
\]
and
\begin{equation}\label{eq:even-potential}
  \Psi(S)=2r(S)+\iota(S).
\end{equation}

The following lemma gives the two properties of this potential used later:
full preimages remain in the class of cyclic intervals, and only an inverse
$a$-step can increase the potential, by at most one.

\begin{lemma}[Even inverse-interval potential]
\label{lem:even-inverse-potential}
Let $n\ge4$ be even, let $a,b$ be the letters of $\mathcal F_n$, and let
$S\subseteq Q_n$ be a cyclic interval in \eqref{eq:even-cyclic-order}.  Then
$Sa^{-1}$ and $Sb^{-1}$ are cyclic intervals in the same order, and
\[
  \Psi(Sb^{-1})\le\Psi(S),
  \qquad
  \Psi(Sa^{-1})\le\Psi(S)+1.
\]
\end{lemma}

\begin{proof}
Traverse the states in the order \eqref{eq:even-cyclic-order}.  Their images
under $a$ and $b$ occur in the respective cyclic lists
\begin{align*}
  a:&\quad 1,1,3,5,\ldots,n-3,0,2,4,\ldots,n-2,\\
  b:&\quad 1,3,5,\ldots,n-1,2,4,\ldots,n-2,1.
\end{align*}
An arithmetic progression in these lists is empty when its first displayed
term exceeds its last; for $n=4$, the lists are $1,1,0,2$ and $1,3,2,1$.
After one of the two cyclically adjacent repetitions is deleted, the first
list is a cyclic shift of
\eqref{eq:even-cyclic-order} with $n-1$ omitted, and the second is such a
shift with $0$ omitted.  Deleting the omitted state from any cyclic interval
leaves a consecutive block in the reduced cyclic order.  The domain positions
whose images lie in this block are consecutive as well.  Restoring the second
of the two adjacent repeated positions either adds one position next to the
block, when the repeated image belongs to $S$, or adds no position.  Thus the
full preimage is a cyclic interval.

In the natural coordinate order $0,1,\ldots,n-1$, the characteristic vectors
of the two preimages are
\begin{align}
  [Sa^{-1}]&=(s_1,s_0,s_1,s_2,\ldots,s_{n-2}),
    \label{eq:even-inverse-a}\\
  [Sb^{-1}]&=(s_1,s_2,\ldots,s_{n-1},s_1).
    \label{eq:even-inverse-b}
\end{align}
Equation \eqref{eq:even-inverse-b} gives
$r(Sb^{-1})=r(S)$ and $\iota(Sb^{-1})=0$.  Hence
$\Psi(Sb^{-1})\le\Psi(S)$.

Equation \eqref{eq:even-inverse-a} gives
\begin{equation}\label{eq:even-r-change}
  r(Sa^{-1})=r(S)+s_0-s_{n-1}.
\end{equation}
For the endpoint pairs
$(s_0,s_{n-1})=(0,0),(0,1),(1,1)$, respectively, the change in
$2r$ is $0,-2,0$, while $\iota$ can increase by at most one.  Therefore
$\Psi(Sa^{-1})-\Psi(S)\le1$ in these three cases.

It remains to consider $(s_0,s_{n-1})=(1,0)$.  Here $\iota(S)=1$.  If
$\iota(Sa^{-1})=1$, then \eqref{eq:even-inverse-a} gives
$s_1=1$ and $s_{n-2}=0$.  The interval $S$ would then contain $0$ and $1$
but omit both $n-2$ and $n-1$.  This is impossible in
\eqref{eq:even-cyclic-order}: one of the two arcs between $0$ and $1$
contains $n-2$, and the other contains $n-1$, while a cyclic interval
containing both endpoints contains at least one of these two arcs.  Hence
$\iota(Sa^{-1})=0$.  Equation \eqref{eq:even-r-change} now gives
$\Psi(Sa^{-1})-\Psi(S)=1$.
\end{proof}

The even-order lower bound now follows by applying one invariant in the
forward action and the other in the reverse full-preimage action.

\begin{proposition}[Even-order lower bound]\label{prop:sharp-even-lower}
If $n\ge4$ is even, then every reset word of $\mathcal F_n$ has length at
least $3n-5$.
\end{proposition}

\begin{proof}
Let $u=g_1g_2\cdots g_L$ be a shortest reset word.  By
\cref{lem:sharp-common-obstruction}, the word $u$ resets $\mathcal F_n$ to
state $1$ and satisfies $\abs u_b\ge n-1$.

For $0\le j\le L$, let $S_j$ be the full preimage of $\{1\}$ under the
suffix of $u$ of length $j$.  Then
\[
  S_0=\{1\},S_1,\ldots,S_L=Q_n,
  \qquad L=\abs u,
\]
and
\[
  S_{j+1}=S_jg_{L-j}^{-1}\quad(0\le j<L).
\]
By \cref{lem:even-inverse-potential}, every $S_j$ is a cyclic interval.  The
endpoint values of the potential are
\[
  \Psi(S_0)=2,
  \qquad
  \Psi(S_L)=2(n-1).
\]
An inverse $b$-step does not increase $\Psi$, and an inverse $a$-step
increases it by at most one.  Telescoping the potential changes gives
\begin{equation}\label{eq:even-a-count}
  \abs u_a\ge \Psi(S_L)-\Psi(S_0)=2n-4.
\end{equation}
Combining \eqref{eq:common-b-count} and \eqref{eq:even-a-count} yields
\[
  \abs u=\abs u_a+\abs u_b
  \ge(2n-4)+(n-1)=3n-5.
\]
Since a shortest reset word has this lower bound, every reset word does.
\end{proof}

\subsection{Odd orders: a middle-layer detour}
\label{subsec:sharp-odd}

Fix an odd integer $n\ge5$.  The letter $b$ is now a permutation.  The common
parity obstruction still gives $\abs u_b\ge n-1$ for every reset word $u$, but
it does not account for the remaining steps in a shortest reset word.  We
instead order the states so that full preimages are cyclic intervals.  Interval
size records progress toward $Q_n$, while the initial position records a
forced detour at the middle size.

Write
\[
  n=2h-1,
  \qquad
  h=\frac{n+1}{2},
\]
and, for $j\in\mathbb Z/n\mathbb Z$, define the state
\[
  y_j=2j\pmod n.
\]
These states occur in the cyclic order
\begin{equation}\label{eq:odd-cyclic-order}
  y_0,y_1,\ldots,y_{n-1}
  =0,2,4,\ldots,n-1,1,3,\ldots,n-2.
\end{equation}
For $i\in\mathbb Z/n\mathbb Z$ and $1\le k<n$, let
\begin{equation}\label{eq:odd-interval}
  I(i,k)=\{y_i,y_{i+1},\ldots,y_{i+k-1}\}.
\end{equation}
Subscripts on $y$ and interval-start indices are read modulo $n$.  For the
boundary sizes, set $I(i,0)=\varnothing$ and $I(i,n)=Q_n$ for every $i$.
The sets $I(i,k)$ with $0\le k\le n$ are called the cyclic intervals in
\eqref{eq:odd-cyclic-order}.  Indices on the reverse-preimage path introduced
later are ordinary integer indices, not residues modulo $n$.

The next lemma translates both inverse-letter actions into changes of the
two interval coordinates.  The two exceptional cases for $a^{-1}$ are the
only steps that change interval size.

\begin{lemma}[Odd inverse-interval dynamics]
\label{lem:odd-inverse-intervals}
Let $n\ge5$ be odd, and let $a,b$ be the letters of $\mathcal F_n$.  For
every $i\in\mathbb Z/n\mathbb Z$ and $1\le k<n$, one has
\begin{equation}\label{eq:odd-b-inverse}
  I(i,k)b^{-1}=I(i-h,k),
\end{equation}
and
\begin{equation}\label{eq:odd-a-inverse}
  I(i,k)a^{-1}=
  \begin{cases}
    I(0,k+1),&\text{if $i=h$},\\
    I(1-k,k-1),&\text{if $i=h-k$},\\
    I(i+h,k),&\text{otherwise}.
  \end{cases}
\end{equation}
\end{lemma}

\begin{proof}
The identity $n=2h-1$ gives $2h\equiv1\pmod n$.  Therefore
\[
  y_j\mathbin{\cdot}b=y_{j+h},
\]
which proves \eqref{eq:odd-b-inverse}.

For the other letter,
\[
  y_0\mathbin{\cdot}a=y_h,
  \qquad
  y_j\mathbin{\cdot}a=y_{j-h}\quad(j\ne0).
\]
Thus, except at the domain point $y_0$, the action is the rotation
$j\mapsto j-h$.  The inverse image under this rotation would be
$I(i+h,k)$.  The actual inverse image gains $y_0$ when the target interval
contains $y_h$ but not its predecessor $y_{h-1}$.  For a proper nonempty
cyclic interval, this condition is equivalent to $i=h$, and adjoining $y_0$
gives $I(0,k+1)$.  The actual inverse image loses $y_0$ when the target
interval contains $y_{h-1}$ but not $y_h$.  This condition is equivalent to
$i=h-k$, and deleting $y_0$ gives $I(1-k,k-1)$.  In every other case, the
presence or absence of $y_0$ agrees with the rotation, which gives the third
line of \eqref{eq:odd-a-inverse}.
\end{proof}

Call an inverse $a$-step \emph{expanding}, \emph{contracting}, or
\emph{neutral} according as it increases, decreases, or preserves the
interval size.  By \cref{lem:odd-inverse-intervals}, an expanding step starts
at $I(h,k)$ and produces $I(0,k+1)$.  In particular, two expanding steps
cannot be consecutive because $h\ne0$ in $\mathbb Z/n\mathbb Z$.

The odd-order lower bound splits the reverse path into three parts.  Below
the middle size, nonconsecutive expansions impose a linear cost.  At the
middle size, one direct coordinate step is forbidden and the path must go
around the other side of an $n$-cycle.  Above the middle size, the remaining
nonconsecutive expansions impose a second linear cost.

\begin{proposition}[Odd-order lower bound]\label{prop:sharp-odd-lower}
If $n\ge5$ is odd, then every reset word of $\mathcal F_n$ has length at
least $3n-5$.
\end{proposition}

\begin{proof}
Let $u=g_1g_2\cdots g_L$ be a shortest reset word.  By
\cref{lem:sharp-common-obstruction}, the word $u$ resets $\mathcal F_n$ to
state $1$.

For $0\le j\le L$, let $S_j$ be the full preimage of $\{1\}$ under the
suffix of $u$ of length $j$.  Thus
\[
  S_0=\{1\}=I(h,1),
  \qquad
  S_L=Q_n,
\]
and
\[
  S_{j+1}=S_jg_{L-j}^{-1}\quad(0\le j<L).
\]
Every $S_j$ is nonempty because the prefix of $u$ of length $L-j$ maps
$Q_n$ into $S_j$.  Moreover, $S_j\ne Q_n$ for $j<L$, since otherwise the
corresponding proper suffix would be a shorter reset word.  Consequently,
\cref{lem:odd-inverse-intervals} applies at every intermediate step and shows
that each $S_j$ is a cyclic interval.

\paragraph{The prefix below the middle layer.}
Let $j_-$ be the largest index for which $\abs{S_{j_-}}<h$.  A single inverse
step changes interval size by at most one.  The maximality of $j_-$ therefore
gives
\[
  \abs{S_{j_-}}=h-1,
  \qquad
  S_{j_-+1}=I(0,h),
\]
and the step to $S_{j_-+1}$ is expanding.  Let $E$ and $D$ be the numbers of
expanding and contracting steps, respectively, among the first $j_-+1$
steps.  Neutral steps do not change size, so
\[
  E-D=\abs{S_{j_-+1}}-\abs{S_0}=h-1.
\]
Hence this prefix contains at least $h-1$ expansions.  Since no two
expansions are consecutive, a sequence containing $h-1$ of them has at least
$2(h-1)-1$ steps.  Thus
\begin{equation}\label{eq:odd-prefix-cost}
  j_-+1\ge2h-3=n-2.
\end{equation}

\paragraph{The middle layer.}
Let $j_+>j_-+1$ be the least index for which $\abs{S_{j_+}}>h$.  Such an
index exists because $S_L=Q_n$.  Every set from $S_{j_-+1}$ through
$S_{j_+-1}$ has size $h$.  A contraction in this segment would give an
index $j>j_-$ with $\abs{S_j}=h-1$, contrary to the maximality of $j_-$.
Therefore every step before the final expansion in this segment is neutral.

Because every common divisor of $h$ and $n=2h-1$ divides $1$, one has
$\gcd(h,n)=1$.  Hence every initial index at size $h$ can be written uniquely
as $i=zh$ with $z\in\mathbb Z/n\mathbb Z$.  On neutral steps,
\cref{lem:odd-inverse-intervals} gives
\[
  b^{-1}:z\longmapsto z-1,
  \qquad
  a^{-1}:z\longmapsto z+1.
\]
The segment starts at $I(0,h)$, so its initial coordinate is $z=0$.  The next
expansion requires $z=1$.  The direct $a^{-1}$-step from $z=0$ to $z=1$ is
not neutral: by \eqref{eq:odd-a-inverse}, it is the contraction
\[
  I(0,h)a^{-1}=I(1-h,h-1).
\]
Thus a neutral walk from $0$ to $1$ cannot use the directed edge $0\to1$ of
the coordinate cycle.  To make the resulting distance bound explicit, lift
the coordinate walk to an integer walk starting at $0$, with each step equal
to $1$ or $-1$.  A lift ending at an integer congruent to $1$ modulo $n$
cannot end at $1+kn$ for $k\ge0$, since reaching any such integer would cross
an edge congruent to $0\to1$.  Its endpoint is therefore at most $1-n$, so
the walk has at least $n-1$ steps.  One additional inverse $a$-step performs
the expansion from size $h$ to size $h+1$.  Hence
\begin{equation}\label{eq:odd-middle-cost}
  j_+-(j_-+1)\ge(n-1)+1=n.
\end{equation}

\paragraph{The suffix above the middle layer.}
The expansion at the end of the middle segment gives
$S_{j_+}=I(0,h+1)$.  Going from size $h+1$ to size
$n=2h-1$ requires a net total of $h-2$ further expansions.  There are
therefore at least $h-2$ expansions in the remaining suffix.  The step that
produces $S_{j_+}$ is itself expanding, and two expanding steps cannot be
consecutive.  Each remaining expansion consequently has an immediately
preceding nonexpanding step within the suffix.  These preceding steps are
distinct, so
\begin{equation}\label{eq:odd-suffix-cost}
  L-j_+\ge2(h-2)=n-3.
\end{equation}

Adding \eqref{eq:odd-prefix-cost}, \eqref{eq:odd-middle-cost}, and
\eqref{eq:odd-suffix-cost} gives
\[
  \abs u=L\ge(n-2)+n+(n-3)=3n-5.
\]
Since $u$ is a shortest reset word, every reset word has at least this length.
\end{proof}

The common upper bound and the two parity-specific lower bounds give the
promised sharpness statement without any computational input.

\begin{corollary}[Sharpness in every order]
\label{cor:refined-bound-sharp}
For every integer $n\ge4$, the automaton $\mathcal F_n$ is strongly connected,
binary, synchronizing, and one-cluster with $m=2$ and $\ell=n-2$.  Moreover,
\[
  \rt(\mathcal F_n)=3n-5
  =(m-1)(n-1)+m\ell.
\]
\end{corollary}

\begin{proof}
Strong connectivity and the one-cluster parameters were established after
\cref{def:sharp-family}.  The upper bound is
\cref{lem:sharp-common-word}.  The matching lower bound is
\cref{prop:sharp-even-lower} when $n$ is even and
\cref{prop:sharp-odd-lower} when $n$ is odd.
\end{proof}

\begin{remark}[The orientation of the long cycle]\label{rem:opposite-orientation}
For odd $n$, replace $b$ in \cref{def:sharp-family} by the inverse cycle
$q\mapsto q-1\pmod n$.  The relabelling $q\mapsto-q\pmod n$ sends the two
letter actions to those of the Wielandt-type automaton $\mathcal W(n,n,2)$.
Indeed, the inverse cycle becomes $r\mapsto r+1\pmod n$, while the other
letter becomes
\[
  0\longmapsto n-1,
  \qquad n-1\longmapsto0,
  \qquad r\longmapsto r+1\quad(1\le r\le n-2),
\]
which are the defining transitions of that automaton
~\cite[Section~2]{GusevPribavkina2015}.  Since $n$ is odd and $n\ge5$, the
integers $2$ and $n$ are coprime and $n>2$.  Therefore
\cite[Theorem~2.2]{GusevPribavkina2015} gives its reset threshold as $2n-3$.
Thus the relative orientation of
the full $n$-cycle and the two-cycle with its tail changes the reset threshold
from $2n-3$ to $3n-5$.
\end{remark}

\bibliographystyle{abbrvurl}
\bibliography{references-2026-07-21}

\end{document}